\pdfoutput=1
\documentclass[a4paper,english,11pt]{article}
\usepackage[margin=1in,footskip=0.3in]{geometry}
\usepackage{times}
\usepackage{microtype}
\usepackage{graphicx}
\usepackage{algorithm}
\usepackage{algorithmic}
\usepackage{amsthm}
\usepackage{amsmath}
\usepackage{tikz}
\usepackage{chngcntr}
\usepackage{caption}
\usepackage{booktabs} % for professional tables
\usepackage[numbers, comma, sort&compress]{natbib}
\usepackage{float}
\usepackage{verbatim}
\usepackage{wrapfig,lipsum,booktabs}
\usepackage{enumerate}
\usepackage{hyperref}

\title{Correlation Clustering with Same-Cluster Queries Bounded by Optimal Cost} %TODO Please add
\date{}
\author{Barna Saha\\
University of California, Berkeley\\
\texttt{barnas@berkeley.edu\thanks{B. Saha is partially supported by an NSF CAREER Award CCF 1652303, a Google Faculty Award and an Alfred P. Sloan fellowship.}} \and
Sanjay Subramanian \\
Allen Institute for Artificial Intelligence\\
\texttt{sanjays@allenai.org}\thanks{Most of this work was completed when the second author was at the University of Pennsylvania and the University of Massachusetts at Amherst. This work was supported in part by the National Science Foundation (NSF) Research Experiences for Undergraduates (REU) program.}}

\newtheorem{thm}{Theorem}[section]
\newtheorem{lem}{Lemma}[section]
\newtheorem{hyp}{Hypothesis}[section]

\newtheorem{cor}{Corollary}[section]

\begin{document}
\setcounter{page}{0}
\maketitle

\begin{abstract}
Several clustering frameworks with interactive (semi-supervised) queries have been studied in the past. Recently, clustering with same-cluster queries has become popular. An algorithm in this setting has access to an oracle with full knowledge of an optimal clustering, and the algorithm can ask the oracle queries of the form, ``Does the optimal clustering put vertices $ u $ and $ v $ in the same cluster?'' Due to its simplicity, this querying model can easily be implemented in real crowd-sourcing platforms and has attracted a lot of recent work. 

In this paper, we study the popular correlation clustering problem (Bansal et al., 2002) under the same-cluster querying framework. Given a complete graph $G=(V,E)$ with positive and negative edge labels, correlation clustering objective aims to compute a graph clustering that minimizes the total number of disagreements, that is the negative intra-cluster edges and positive inter-cluster edges. In a recent work, Ailon et al. (2018b) provided an approximation algorithm for correlation clustering that approximates the correlation clustering objective within $(1+\epsilon)$ with $O(\frac{k^{14}\log{n}\log{k}}{\epsilon^6})$ queries when the number of clusters, $k$, is fixed. For many applications, $k$ is not fixed and can grow with $|V|$. Moreover, the dependency of $k^14$ on query complexity renders the algorithm impractical even for datasets with small values of $k$. 

In this paper, we take a different approach. Let $ C_{OPT} $ be the number of disagreements made by the optimal clustering. We present algorithms for correlation clustering whose error and query bounds are parameterized by $C_{OPT}$  rather than by the number of clusters. Indeed, a good clustering must have small $C_{OPT}$. Specifically, we present an efficient algorithm that recovers an exact optimal clustering using at most $2C_{OPT} $ queries and an efficient algorithm that outputs a $2$-approximation using at most $C_{OPT} $ queries. In addition, we show under a plausible complexity assumption, there does not exist any polynomial time algorithm that has an approximation ratio better than $1+\alpha$ for an absolute constant $\alpha >0$ with $o(C_{OPT})$ queries. Therefore, our first algorithm achieves the optimal query bound within a factor of $2$.

We extensively evaluate our methods on several synthetic and real-world datasets using real crowd-sourced oracles. Moreover, we compare our approach against known correlation clustering algorithms that do not perform querying. In all cases, our algorithms exhibit superior performance.

\end{abstract}
\newpage
\setcounter{page}{1}
\section{Introduction}
\label{introduction}
In correlation clustering, the algorithm is given potentially inconsistent information about similarities and dissimilarities between pairs of vertices in a graph, and the task is to cluster the vertices so as to minimize disagreements with the given information \cite{Bansal-02, Chawla-15}. The correlation clustering problem was first proposed by Bansal, Blum and Chawla \cite{Bansal-02} and since then it has found numerous applications in document clustering, image segmentation, grouping gene expressions etc. \cite{Bansal-02, Chawla-15}.

In correlation clustering, we are given a complete graph $G=(V,E)$, $|V|=n$, where each edge is labelled either $+$ or $-$. An optimal clustering partitions the vertices such that the number of intra-cluster negative edges and inter-cluster positive edges is minimized. The problem is known to be NP-Hard. The seminal work of Bansal et al. \cite{Bansal-02} gave a constant factor approximation for correlation clustering. Following a long series of works \cite{Bansal-02,Charikar-05,Ailon-08,Giotis-06,Demaine-06}, the best known approximation bounds till date are a 3-approximation combinatorial algorithm \cite{Ailon-05} and a $2.06$-approximation based on linear programming rounding \cite{Chawla-15}. The proposed linear programming relaxation for correlation clustering  \cite{Charikar-05,Ailon-05,Chawla-15} is known to have an integrality gap of $2$, but there does not exist yet a matching algorithm that has an approximation ratio $2$ or lower. 

Correlation clustering problem can be extended to weighted graphs for an $O(\log{n})$-approximation bound and is known to be optimal \cite{Demaine-06}. Moreover, when one is interested in maximizing agreements, a polynomial time approximation scheme was provided by Bansal et al. \cite{Bansal-02}.

Over the last two decades, crowdsourcing has become a widely used way to generate labeled data for supervised learning. The same platforms that are used for this purpose can also be used for unsupervised problems, thus converting the problems to a semi-supervised active learning setting. This can often lead to significant improvements in accuracy.  However, using crowdsourcing introduces another dimension to the optimization problems, namely minimizing the amount of crowdsourcing that is used. The setting of active querying has been studied previously in the context of various clustering problems. Balcan and Blum \cite{Balcan-08} study a clustering problem in which the only information given to the algorithm is provided through an oracle that tells the algorithm either to ``merge'' two clusters or to ``split'' a cluster. More recently, Ashtiani, Kushgra and Ben-David \cite{Ashtiani-16} considered a framework of same-cluster queries for clustering; in this framework, the algorithm can access an oracle that has full knowledge of an optimal clustering and can issue queries to the oracle of the form ``Does the optimal clustering put vertices $u$ and $v$ in the same cluster?'' Because of its simplicity, such queries are highly suited for crowdsourcing and has been studied extensively both in theory community \cite{Ailon-18a,Mazumdar-17,Ailon-18b,GHS18} and in applied domains \cite{verroios2017waldo,firmani2018robust,gruenheid2015fault,verroios2015entity}. Correlation clustering has also been considered in this context. Ailon, Bhattacharya and Jaiswal \cite{Ailon-18b} study correlation clustering in this framework under the assumption that the number $ k $ of clusters is fixed. They gave an $(1+\epsilon)$ approximation algorithm for correlation clustering that runs in polynomial time and issues $O(k^{14}\log{n}\log{k}/\epsilon^6)$ queries. However, for most relevant applications, the number of clusters $k$ is not fixed. Even for fixed $k$, the dependence of $k^{14}$ is huge (consider $k=2$ and $2^{14}=16384$ with additional constants terms hidden under $O()$ notation). 

In this paper, we give near-optimal algorithms for correlation clustering with same-cluster queries that are highly suitable for practical implementation and whose performance is parameterized by the optimum number of disagreements. Along with providing theoretical guarantees, we perform extensive experiments on multiple synthetic and real datasets. Let $ C_{OPT} $ be the number of disagreements made by the optimal clustering. Our contributions are as follows.
\begin{enumerate}
	\item A deterministic algorithm that outputs an {\em optimal} clustering using at most $2C_{OPT} $ queries (Section~\ref{sec:opt}).
	\item An expected $2$-approximation algorithm that uses at most $C_{OPT}$ queries in expectation (Section~\ref{sec:approx}). 
	\item A new lower bound that shows it is not possible to get an $(1+\alpha)$ approximation for some constant $\alpha >0$ with any polynomial time algorithm that issues $o(C_{OPT})$ queries assuming GAP-ETH (see definition in Section~\ref{sec:lower}).
	\item An extensive experimental comparison that not only compares the effectiveness of our algorithms, but also compares the state-of-the art correlation clustering algorithms that do not require any querying (Section~\ref{sec:experiment}).
\end{enumerate}
Assumption of an optimum oracle \cite{Ashtiani-16, Ailon-18b} is quite strong in practice. However, our experiments reveal that such an assumption is not required. In correlation clustering, often the $\pm$ edges are generated by fitting an automated classifier, where each vertex corresponds to some object and is associated with a feature vector. In our experiments with real-world data, instead of an optimum oracle, we use crowdsourcing. By making only a few pair-wise queries to a crowd oracle, we show it is possible to obtain an optimum or close to optimum clustering. After our work, it came to our notice that it may be possible to use Bocker et al.'s~\cite{Bocker2009} results on fixed-parameter tracktability of cluster editing to adapt to our setting, and get better constants on the query complexity. In this long version of the paper, we include experimental results for the branching algorithm of \cite{Bocker2009} with original running time $ O(1.82^k+n^3) $, adapted to our setting. Our algorithms and techniques are vastly different from \cite{Bocker2009} and are also considerably simpler.

\section{Related Work}
Asthiani et al. \cite{Ashtiani-16} considered the $ k $-means objective with same-cluster queries and showed that it is possible to recover the optimal clustering under $k$-means objective with high probability by issuing $ O(k^2\log k + k\log n) $ queries if a certain margin condition holds for each cluster. Gamlath, Huang and Svensson extended the above result when approximation is allowed \cite{GHS18}. 
Ailon et al. \cite{Ailon-18b} studied correlation clustering with same-cluster queries and showed that there exists an $(1+\epsilon)$ approximation for correlation clustering where the number of queries is a (large) polynomial in $k$. Our algorithms are different from those in \cite{Ailon-18b} in that our guarantees are parameterized by $ C_{OPT} $ rather than by $ k $. Kushagra et al.~\cite{Kushagra-18} study a restricted version of correlation clustering where the valid clusterings are provided by a set of hierarchical trees and provide an algorithm using same-cluster queries for a related setting, giving guarantees in terms of the size of the input instance (or the VC dimension of the input instance) rather than $C_{OPT} $. \cite{Mazumdar-17} studied, among other clustering problems, a random instance of correlation clustering under same-cluster queries.

Our algorithms are based on the basic 3-approximation algorithm of Ailon et al.~\cite{Ailon-05} that selects a pivot vertex randomly and forms a cluster from that vertex and all of its $ + $-neighbors. They further honed this approach by choosing to keep each vertex in the pivot's cluster with a probability that is a function of the linear programming solution. Chawla et al.~\cite{Chawla-15} used a more sophisticated function of the linear programming solution to design the current state-of-the-art algorithm, which gives a 2.06 approximation for correlation clustering.

\section{Finding an Optimal Clustering}
\label{sec:opt}

We are given a query access to an oracle that given any two vertices $u$ and $v$ returns whether or not $u$ and $v$ are together in a cluster in an optimal solution. Let $OPT$ denote the optimal solution which is used by the oracle. Given a positive ($+$) edge $(u,v)$, if $OPT$ puts $u$ and $v$ in different clusters, then we say $OPT$ makes a mistake on that edge. Similarly for a negative ($-$) edge $(u,v)$, if $OPT$ puts them together in a cluster then again $OPT$ makes a mistake on it. Similarly, our algorithm can decide to make mistakes on certain edges and our goal is to minimize the overall number of mistakes. It is easy to see that an optimal solution for a given input graph makes mistakes only on edges that are part of a $(+, +, -)$ triangle. Moreover, any optimal solution must make at least one mistake in such a triangle.

The pseudocode for our algorithm, \textsc{QueryPivot}, is given in Algorithm~\ref{alg:exact}. The algorithm is as follows (in the following description, we give in brackets the corresponding line number for each step). We pick a pivot $u$ arbitrarily from the set of vertices that are not clustered yet [line 5]. For each $(+, +, -)$ triangle $(u, v, w)$ [line 10], if we have not yet determined via queries that $OPT$ makes a mistake on $\{u, v\}$ or that $OPT$ makes a mistake on $\{u, w\}$ [lines 11-14], then (1) we query $\{u, v\}$ [line 17] and if $ OPT $ makes a mistake on this edge, we too decide to make a mistake on this edge and proceed to the next $(+, +, -)$ triangle involving $u$ and (2) if $OPT$ does not make a mistake on $\{u, v\}$, then we query $\{u, w\}$ [line 23] and make a mistake on it if $OPT$ makes a mistake on it. Note that if we have already queries one of $\{u, v\}$ or $\{u, w\}$ and found a mistake, we do not query the other edge [line 11]. Once we have gone through all $(+, +, -)$ triangles involving $u$ then for every $v \neq u$, if we have not already decided to make a mistake on $\{u, v\}$, then if $ \{u, v\} $ is a $+$ edge we keep $v$ in $u$'s cluster and if $\{u, v\}$ is a $ - $ edge we do not put $ v $ in $ u $'s cluster. On the other hand, if we have decided to make a mistake on $ \{u, v\} $, then if $ \{u, v\} $ is a $ - $ edge we keep $ v $ in $ u $'s cluster and if $ \{u, v\} $ is a $ + $ edge we do not put $ v $ in $ u $'s cluster. Finally, we remove all vertices in $u$'s cluster from the set of remaining vertices and recursively call the function on the set of remaining vertices.

In the pseudocode, $Queried[v]=1$ means the algorithm has already issued a query $(pivot,v)$ to the oracle, $Mistake[v]=1$ means it has decided to make a mistake on the edge $(pivot,v)$ based on the oracle answer, and $Oracle(pivot, v)$ returns $ 1 $ iff $OPT$ makes a mistake on the edge $\{pivot, v\}$. We prove the following theorem that shows that \textsc{QueryPivot} is able to recover the optimal clustering known to the oracle with a number of queries bounded in terms of $ C_{OPT} $.
\begin{thm}
		\label{thm:querypivot}
		Let $ C_{OPT} $ be the number of mistakes made by an optimal clustering. The \textsc{QueryPivot} algorithm makes $ C_{OPT} $ mistakes and makes at most $ 2C_{OPT} $ queries to the oracle.
\end{thm}

\begin{algorithm}[tb]
	\caption{\textsc{QueryPivot}}
	\label{alg:exact}
	\begin{algorithmic}[1]
		\STATE {\bfseries Input:} vertex set $V$, adjacency matrix $A$, oracle $ \textit{Oracle} $
		\IF{$ V == \emptyset $}
		\STATE {\bfseries return} $ \emptyset $
		\ENDIF
		\STATE $ \textit{pivot}  \gets $ Arbitrary vertex in $ V $
		\STATE $ \textit{T} \gets $ all $  (+, +, -)  $ triangles that include \textit{pivot}
		\STATE $ \textit{C} \gets \textit{V} $
		\STATE $ \textit{Queried} \gets $ length-n array of zeros
		\STATE $ \textit{Mistakes} \gets $ length-n array of zeros
		\FOR{$ (\textit{pivot}, v, w) \in T $}
		\IF{$ \textit{Mistake}[v] == 1 $ {\bf or} $ \textit{Mistake}[w] == 1$}
		\STATE continue
		\ELSIF{$ \textit{Queried}[v] == 1 $ {\bf and} $ \textit{Queried}[w] == 1 $}
		\STATE continue
		\ELSIF{$ \textit{Queried}[v] == 0 $}
		\STATE $ \textit{Queried}[v] \gets 1 $
		\IF{$ \textit{Oracle}(\textit{pivot}, v) == 1 $}
		\STATE $ \textit{Mistake}[v] \gets 1 $
		\ENDIF
		\ENDIF
		\IF{$ \textit{Queried}[w] == 0 $ {\bf and} $ \textit{Mistake}[v] == 0 $}
		\STATE $ \textit{Queried}[w] \gets 1 $
		\IF{$ \textit{Oracle}(\textit{pivot}, w) == 1 $}
		\STATE $ \textit{Mistake}[w] \gets 1 $
		\ENDIF
		\ENDIF
		\ENDFOR
		\FOR{$ v \in V \setminus \{\textit{pivot}\} $}
		\IF{$ (v \in N^-(\textit{pivot}) $ {\bf and} $ \textit{Mistake}[v] == 0) $ {\bf or} $ (v \in N^+(\textit{pivot}) $ {\bf and} $ \textit{Mistake}[v] == 1) $}
		\STATE $ C = C \setminus \{v\} $
		\ENDIF
		\ENDFOR
		\STATE {\bfseries return} $ \{C\} \cup \textsc{QueryPivot}(V \setminus C, A, \textit{Oracle}) $
	\end{algorithmic}
\end{algorithm}

For a given cluster $ C $ and a vertex $ w \in C $, we denote by $ N_C^+(w) $ the set of vertices in $ C $ that have $ + $ edges with $ w $. Similarly, we denote by $ N_C^-(w) $ the set of vertices in $ C $ that have $ - $ edges.

The algorithm time complexity is dominated by the time taken to check $(+,+,-)$ triangles involved with the pivots. Let $E^+$ denote the set of positive edges in $G$. Then all the $(+,+,-)$ triangles that include a pivot can be checked in time $O(|E^+|*n)$. 

\begin{lem}
	The \textsc{QueryPivot} algorithm outputs a valid partition of the vertices.
\end{lem}
\begin{proof}
	Note that the pivot is never removed from $ C $. Hence, between each pair of consecutive recursive calls, at least one vertex is removed from $ V $. The algorithm must then terminate after at most $ n $ recursive calls. Moreover, in each recursive call, the set of vertices passed to the next recursive call is disjoint from the cluster created in that recursive call. Thus, inductively, the sets returned by the algorithm must be disjoint.
\end{proof}

\begin{lem}
	\label{lem:suboptimal1}
	Consider a clustering $\mathcal{C}$ in which some cluster $C$ contains vertices $u, v$ s.t. $\{u, v\}$ is a $ - $ edge and s.t. $u$ and $v$ do not form a $(+, +, -)$ triangle with any other vertex in $C$. $\mathcal{C} $ is suboptimal.
\end{lem}
\begin{proof}
	If we were to remove $ v $ from $ C $ and put it in a singleton cluster, we would make $ |N_C^-(v)|-|N_C^+(v)| $ fewer mistakes than $ \mathcal{C} $. If $ |N_C^-(v)| - |N_C^+(v)| > 0 $, then $ \mathcal{C} $ is suboptimal. Therefore, assume $ |N_C^-(v)| \leq |N_C^+(v)| $.  Now, note that $ \forall w \in N_C^+(v) $, $ w \in N_C^-(u) $ because otherwise, $ u $, $ v $, and $ w $ form a $ (+, +, -) $ triangle. Thus $|N_C^+(v)| \leq |N_C^-(u)|-1$ because $ v \in N_C^-(u) $. Moreover, $ |N_C^+(u)| \leq |N_C^-(v)| - 1 $ because $ u \in N_C^-(v) $.
	
	Hence, if we were to remove $ u $ from $ C $ and put it in a singleton cluster, we would make $ |N_C^-(u)|-|N_C^+(u)| \geq |N_C^+(v)| - |N_C^-(v)| + 2 $ fewer mistakes than $ \mathcal{C} $. Since $ |N_C^-(v)| - |N_C^+(v)| \leq 0 $, $ |N_C^+(v)| - |N_C^-(v)| + 2 > 0 $, so $ \mathcal{C} $ is suboptimal.
\end{proof}

\begin{lem}
	\label{lem:suboptimal2}
	Consider a clustering $ \mathcal{C} $ in which a cluster $ C_1 $ contains a vertex $ u $, a different cluster $ C_2 $ contains a vertex $ v $, $ \{u, v\} $ is a $ + $ edge, and in every $ (+, +, -) $ triangle that includes $ \{u, v\} $, the clustering makes at least $ 2 $ edge mistakes. $ \mathcal{C} $ is suboptimal.
\end{lem}
\begin{proof}
	If we were to remove $ u $ from $ C_1 $ and put it in $ C_2 $, we would make $ |N_{C_2}^+(u)| + |N_{C_1}^-(u)| - |N_{C_2}^-(u)| - |N_{C_1}^+(u)| = 2|N_{C_2}^+(u)| + |C_1| - |C_2| - 2|N_{C_1}^+(u)| $ fewer mistakes than $ \mathcal{C} $. If $ 2|N_{C_2}^+(u)| + |C_1| - |C_2| - 2|N_{C_1}^+(u)| > 0 $, then $ \mathcal{C} $ is suboptimal. Otherwise, note that $ \forall w \in N_{C_1}^+(u) $, $ w \in N_{C_1}^+(v) $ because if not, $ u, v, w $ would form a $ (+, +, -) $ triangle in which the algorithm makes fewer than $ 2 $ edge mistakes. By a similar argument, $ \forall w \in N_{C_2}^+(v) $, $ w \in N_{C_2}^+(u) $. Thus, since in addition, $ \{u, v\} $ is a $ + $ edge, we have that $ |N_{C_1}^+(v)| \geq |N_{C_1}^+(u)|+1 $ and $ |N_{C_2}^+(u)| \geq |N_{C_2}^+(v)|+1 $. Now if we were to remove $ v $ from $ C_2 $ and put it in $ C_1 $, the number of mistakes will reduce by $ |N_{C_1}^+(v)|+|N_{C_2}^-(v)| - |N_{C_1}^-(v)| - |N_{C_2}^+(v)| = 2|N_{C_1}^+(v)| + |C_2| - |C_1| - 2|N_{C_2}^+(v)| $. Since $ |N_{C_1}^+(v)| \geq |N_{C_1}^+(u)|+1 $ and $ |N_{C_2}^+(v)| \leq |N_{C_2}^+(u)|-1 $, we have that
	$ 2|N_{C_1}^+(v)| + |C_2| - |C_1| - 2|N_{C_2}^+(v)| \geq 2(|N_{C_1}^+(u)|+1) + |C_2| - |C_1| - 2(|N_{C_2}^+(u)-1) $
	. Since $ 2|N_{C_2}^+(u)| + |C_1| - |C_2| - 2|N_{C_1}^+(u)| \leq 0 $, $ 2(|N_{C_1}^+(u)|+1) + |C_2| - |C_1| - 2(|N_{C_2}^+(u)-1) > 0 $, so $ \mathcal{C} $ is suboptimal.
\end{proof}

\begin{lem}
	\label{lem:querypivot_correctness}
	When given an oracle corresponding to an optimal clustering $ OPT $, the clustering returned by the \textsc{QueryPivot} algorithm is identical to $ OPT $. It follows that the algorithm's clustering makes at most as many mistakes as $ OPT $.
\end{lem}
\begin{proof}
	We will prove inductively that in each recursive call, the cluster $ C $ returned by the algorithm is a cluster in $ OPT $. Note that at the beginning of the first recursive call, the claim that all clusters formed so far are clusters in $ OPT $ is vacuously true because there are no clusters yet formed. Now consider an arbitrary but particular recursive call, and let $ u $ be the pivot in this recursive call. Suppose for contradiction that $ C $ is not a cluster in $ OPT $.\\
	{\it Case 1}: There is a vertex $ v $ such that $ v \notin C $, but in $ OPT $, $ v $ is in the same cluster as $ u $. Let $ H $ be the cluster in $ OPT $ that contains $ u $ and $ v $. First, observe that $ H $ must be a subset of the remaining vertices in this recursive call; otherwise, one of the clusters formed in a previous call contains some vertex in $ H $ but does not include $ u $, contradicting the induction hypothesis because this previously formed cluster is not a cluster in $ OPT $. Next, note that for any mistake that the algorithm makes on an edge incident on a pivot, the algorithm queries the $ OPT $ oracle and makes the mistake iff $ OPT $ makes the mistake. Then if $ \{u, v\} $ is a $ + $ edge, then the algorithm must have queried the oracle for $ \{u, v\} $ and found that $OPT$ makes a mistake on it because the algorithm decided to make a mistake on that edge. This implies that $ OPT $ puts $ u $ and $ v $ in different clusters, which is a contradiction. Now suppose instead that $ \{u, v\} $ is a $ - $ edge. Again if the algorithm queried the oracle for $ \{u, v\} $, then $ OPT $ must have put $ u $ and $ v $ in different clusters, so it must be the case that the algorithm did not query the oracle for $ \{u, v\} $. It follows that for any $ (+, +, -) $ triangle $ (u, v, w) $ that includes $ \{u, v\} $, our algorithm has queried $\{u,w\}$ and found $ OPT $ makes a mistake on the $ + $ edge $ \{u, w\} $. Then for any such triangle, $ w \notin H $. It follows that $ u $ and $ v $ do not form a $ (+, +, -) $ triangle with any vertex in $ H $. Since $ u $ and $ v $ are in the same cluster $ H $ in $ OPT $, $ \{u, v\} $ is a $ - $ edge, and $ u $ and $ v $ do not form a $ (+, +, -) $ triangle with any other vertex in $ H $, the conditions for Lemma ~\ref{lem:suboptimal1} are satisfied. Therefore, $ OPT $ is a suboptimal clustering, which is a contradiction.\\
	{\it Case 2}: There is a vertex $ v $ such that $ v \in C $, but in $ OPT $, $ v $ is not in $ u $'s cluster. As in the first case, if $ \{u, v\} $ were a $ - $ edge, the algorithm must make a mistake on $ \{u, v\} $ and so must have queried $ OPT $ and found that $ OPT $ made a mistake on $ \{u, v\} $, a contradiction. Now suppose instead that $ \{u, v\} $ is a $ + $ edge. If there is some vertex $ w $ that was clustered prior to this recursive call s.t. $ \{u, v, w\} $ is a $ (+, +, -) $ triangle in which $ OPT $ makes exactly one mistake (on $ \{u, v\} $), then note that either $ u $ or $ v $ should be in the same cluster as $ w $ because one of $ \{u, w\} $ and $ \{v, w\} $ must be a $ + $ edge; in this case, we have reached a contradiction with the inductive hypothesis because the previously formed cluster that included $ w $ did not include $ u $ or $ v $. Then in order to show that the conditions for Lemma ~\ref{lem:suboptimal2} are satisfied, we must show that for every vertex $ w $ in the set of remaining vertices when $ u $ is the pivot, if $ (u, v, w) $ is a $ (+, +, -) $ triangle, then $ OPT $ must make at least two mistakes in the triangle. Since $ OPT $ makes a mistake on $ \{u, v\} $ but the algorithm does not do so, it must be the case that the algorithm did not query $ \{u, v\} $. Since the algorithm did not query $ \{u, v\} $, for every $ (+, +, -) $ triangle $ (u, v, w) $ that includes $ \{u, v\} $ and such that $ w $ is in the set of remaining vertices when $ u $ is the pivot, $ OPT $ must make a mistake on $ \{u, w\} $. Then since $ OPT $ makes a mistake on $ \{u, v\} $ and on $ \{u, w\} $ in any $ (+, +, -) $ triangle $ (u, v, w) $, we have by Lemma ~\ref{lem:suboptimal2} that $ OPT $ is suboptimal clustering, which is a contradiction.
\end{proof}
\begin{lem}
	\label{lem:querypivot_queries}
	Let $ C_{OPT} $ be the number of mistakes made by an optimal clustering $ OPT $. Then the \textsc{QueryPivot} algorithm makes at most $ 2C_{OPT} $ queries to the oracle.
\end{lem}
\begin{proof}
	The algorithm queries the oracle only when considering $ (+, +, -) $ triangles. Note that whenever considering a particular $ (+, +, -) $ triangle, if the algorithm makes a query, it makes at most two queries when considering that triangle and makes at least one mistake that had not been made when considering previous triangles. Therefore, the algorithm makes at most twice as many queries as mistakes. Since the algorithm makes exactly $ C_{OPT} $ mistakes, the algorithm makes at most $ 2C_{OPT} $ queries.
\end{proof}
Theorem ~\ref{thm:querypivot} follows directly from Lemmas ~\ref{lem:querypivot_correctness} and ~\ref{lem:querypivot_queries}.

\section{A 2-Approximation Algorithm for Correlation Clustering}
\label{sec:approx}
A natural question that arises from \textsc{QueryPivot} is how to use fewer queries and obtain an approximation guarantee that is better than the state-of-the-art outside the setting with same-cluster queries, which is a $ 2.06 $-approximation. In this section, we show that a randomized version of \textsc{QueryPivot} gives a $ 2 $-approximation in expectation using at most $ C_{OPT} $ queries in expectation.

The algorithm \textsc{RandomQueryPivot}$ (p) $ is as follows. We pick a pivot $ u $ uniformly at random from the vertices yet to be clustered. For each $ (+, +, -) $ triangle $ (u, v, w) $, we have two cases. (1) If $ \{u, v\} $ and $ \{u, w\} $ are both $ + $ edges, then with probability $ p $ (chosen appropriately), we query both $ \{u, v\} $ and $ \{u, w\} $ and for each of these two edges we make a mistake on the edge iff $ OPT $ makes a mistake on the edge. With probability $ 1-p $ we make no queries for this triangle and proceed to the next triangle.
(2) If one of $ \{u, v\} $ and $ \{u, w\} $ is a $ + $ edge and the other is a $ - $ edge, then with probability $ p $, we do the following. First, we query the $ + $ edge and if $ OPT $ makes a mistake on it, then we make a mistake on it and proceed to the next triangle. If $ OPT $ does not make a mistake on the $ + $ edge, then we query the $ - $ edge and make a mistake on the $ - $ edge iff $ OPT $ does so. Again, with probability $ 1-p $ we make no queries for this triangle and proceed to the next triangle.
Once we have gone through all triangles, if we have not already decided to make a mistake on $ \{u, v\} $, then if $ \{u, v\} $ is a $ + $ edge we keep $ v $ in $ u $'s cluster and if $ \{u, v\} $ is a $ - $ edge we do not put $ v $ in $ u $'s cluster. On the other hand, if we have decided to make a mistake on $ \{u, v\} $, then if $ \{u, v\} $ is a $ - $ edge we keep $ v $ in $ u $'s cluster and if $ \{u, v\} $ is a $ + $ edge we do not put $ v $ in $ u $'s cluster. Finally, we remove all vertices in $ u $'s cluster from the set of remaining vertices and recursively call the function on the set of remaining vertices. Note that given a pivot $u$ and a $(+,+,-)$ triangle containing $u$, if the algorithm chooses not to query either of the edges incident on $u$, then the algorithm must make a mistake on the edge opposite to $u$ in that triangle.

\begin{algorithm}[tb]
	\caption{\textsc{RandomQueryPivot}}
	\label{alg:approx}
	\begin{algorithmic}[1]
		\STATE {\bfseries Input:} vertex set $V$, adjacency matrix $A$, oracle $ \textit{Oracle} $, parameter $ p $
		\IF{$ V == \emptyset $}
		\STATE {\bfseries return} $ \emptyset $
		\ENDIF
		\STATE $ \textit{pivot}  \gets $ Random vertex in $ V $
		\STATE $ \textit{T} \gets $ all $  (+, +, -)  $ triangles that include \textit{pivot}
		\STATE $ \textit{C} \gets \textit{V} $
		\STATE $ \textit{Queried} \gets $ length-n array of zeros
		\STATE $ \textit{Mistakes} \gets $ length-n array of zeros
		\FOR{$ (\textit{pivot}, v, w) \in T $}
		\STATE // Without loss of generality, suppose that $ \{\textit{pivot}, v\} $ is a $ + $ edge
		\STATE Sample $ r $ from $ Uniform(0, 1) $
		\IF{$ r > p $}
		\STATE continue
		\ENDIF
		\IF{$ \textit{Oracle}(\textit{pivot}, v) == 1 $}
		\STATE $ \textit{Mistake}[v] \gets 1 $
		\ENDIF
		\IF{$ Mistake[v] == 0 $ {\bf or} $ \{\textit{pivot}, w\} $ is a $ + $ edge}
		\IF{$ \textit{Oracle}(\textit{pivot}, w) == 1 $}
		\STATE $ \textit{Mistake}[v] \gets 1 $
		\ENDIF
		\ENDIF
		\ENDFOR
		\FOR{$ v \in V \setminus \{\textit{pivot}\} $}
		\IF{$ (v \in N^-(\textit{pivot}) $ {\bf and} $ \textit{Mistake}[v] == 0) $ {\bf or} $ (v \in N^+(\textit{pivot}) $ {\bf and} $ \textit{Mistake}[v] == 1) $}
		\STATE $ C = C \setminus \{v\} $
		\ENDIF
		\ENDFOR
		\STATE {\bfseries return} $ \{C\} \cup \textsc{RandomQueryPivot}(V \setminus C, A, \textit{Oracle}) $
	\end{algorithmic}
\end{algorithm}
\begin{thm}
	\label{thm:randomquerypivot}
	\textsc{RandomQueryPivot}$ (p) $ gives a $ \max\left(2, \frac{3}{1+2p}\right) $-approximation in expectation and uses at most $ \max(4p, 1)*C_{OPT} $ queries in expectation.
\end{thm}
\begin{cor}
	When $ p = 0.25 $, \textsc{RandomQueryPivot} gives a $ 2 $-approximation in expectation and uses at most $ C_{OPT} $ queries in expectation.
\end{cor}
\begin{lem}
    \label{lemma:queryequalprob}
    In an arbitrary but particular recursive call, the probability that \textsc{RandomQueryPivot} queries edge $ \{u, v\} $ on which $ OPT $ makes a mistake given that $ u $ is the pivot is equal to the probability that \textsc{RandomQueryPivot} queries edge $ \{u, v\} $ given that $ v $ is the pivot.
\end{lem}
\begin{proof}
    For a $ + $ edge $ \{u, v\} $ on which $ OPT $ makes a mistake, the probability that the edge is queried given that one of the vertices is the pivot is a function only of the number of $ (+, +, -) $ triangles that include the edge. In particular, if $ T $ is the number of $ (+, +, -) $ triangles including the edge, the probability that the edge is queried is $ 1-(1-p)^T $. This number of triangles does not depend on the pivot vertex, so the claim holds if $ \{u, v\} $ is a $ + $ edge. If $ \{u, v\} $ is a $ - $ edge, then we claim that the probability that $ \{u, v\} $ is queried given that either $ u $ or $ v $ is a function only of the number of $ (+, +, -) $ triangles that include $ \{u, v\} $ in which $ OPT $ makes a mistake only on this $ - $ edge. This claim is true because (1) in any $ (+, +, -) $ triangle in which $ OPT $ makes a mistake on the $ - $ and a $ + $ edge, $ OPT $ must make a mistake on all of the three edges in the triangle and (2) when considering a $ (+, +, -) $ triangle such that the pivot is an endpoint of the $ - $ edge, the algorithm queries the $ - $ edge iff $ OPT $ does not make a mistake on the $ + $ edge of which the pivot is an endpoint. It follows that for any $ (+, +, -) $ triangle in which the algorithm queries the $ - $ edge, $ OPT $ must make a mistake only on the $ - $ edge. Since the number of $ (+, +, -) $ triangles that include $ \{u, v\} $ in which $ OPT $ makes a mistake only on $ \{u, v\} $ does not depend on whether $ u $ or $ v $ is the pivot, the claim holds when $ \{u, v\} $ is a $ - $ edge.
\end{proof}
Let $ s_{uv} = 1 $ if $ \{u, v\} $ is a $ - $ edge and $ 0 $ otherwise. Let $ c^*_{uv} $ equal $ 1 $ if $ OPT $ makes a mistake on $ \{u, v\} $ and $ 0 $ otherwise. 

Let $ OPT^t $ be the number of edges $ \{u, v\} $ s.t. $ c^*_{uv} = 1 $ and the algorithm makes a decision on $ \{u, v\} $ in iteration $ t $. Let $ ALG^t $ be the number of edges $ \{u, v\} $ s.t. the algorithm makes a mistake on $ \{u, v\} $ and the algorithm makes a decision on $ \{u, v\} $ in iteration $ t $. 

Let $ V_t $ be the set of vertices remaining at the beginning of iteration $ t $. Let $ D^t_{uv} $ be the event that the algorithm makes a decision on $ \{u, v\} $ in iteration $ t $.
\begin{lem}
	\label{lemma:loop}
	%(see proof of Lemma 4 in \cite{Chawla-15}) 
	Let $ T $ be the number of iterations that the algorithm takes to cluster all vertices. If $ E[ALG^t | V_t] \leq \alpha E[OPT^t | V_t] $, for each iteration $ t $, then $ E\left[\sum_{t=1}^T ALG^t\right] \leq \alpha E\left[\sum_{t=1}^T OPT^t \right] $.
\end{lem}
\begin{proof}
	Define $ X_0 = 0 $ and for each $ s > 0 $, define $ X_s = \sum_{t=1}^s \alpha OPT^t - ALG^t $. If the condition in the lemma holds, then $ X_s $ is a submartingale because $ E[X_{s+1}|X_s] \geq X_s $. Also, $ T $ is a stopping time that is almost surely bounded (since $ T \leq n $ with probability $ 1 $). By Doob's optional stopping theorem \citep[p.~100]{williams1991probability}, if $ T $ is a stopping time that is almost surely bounded and $ X $ is a discrete-time submartingale, then $ E[X_T] \geq E[X_0] $. Then we have that $ E[X_T] = E\left[\sum_{t=1}^T \alpha OPT^t - ALG^t\right] \geq E[X_0] = 0 $.
\end{proof}
\begin{lem}
	\label{lem:randomquery_mistake}
	The expected number of mistakes made by the algorithm's clustering is at most $ \max\left(2, \frac{3}{1+2p}\right)C_{OPT} $.
\end{lem}
\begin{proofsketch}
The full proof of Lemma~\ref{lem:randomquery_mistake} is given in the supplementary material. Here we give a sketch of the proof.
	By Lemma ~\ref{lemma:loop}, if we show that $ E[ALG^t - \alpha OPT^t] \leq 0 $ for any $ t $, where $ \alpha \leq \max\left(2, \frac{3}{1+2p}\right)$, then the claim will follow.  Let $ A^t_w $ be the event that $ w \in V_t $ is the pivot in iteration $ t $.
	\begin{align*}
		E[OPT^t | V_t] =&\; \sum_{\{u, v\} \subseteq E \cap (V_t \times V_t)} \frac{c^*_{uv}}{|V_t|} \sum_{w \in V_t} \Pr[D^t_{uv}|A^t_w]
	\end{align*}
	Now we will write $ E[ALG^t|V_t] $ by charging the algorithm's mistakes to each of $ OPT $'s mistakes. 
	
	Let $ M^t_{uv} $ be the charge incurred to $ \{u, v\} $ in iteration $ t $. We will assign charges such that $M^t_{uv}=0$ if $c^*_{uv}=0$. Then
	\begin{align*}
		E[ALG^t|V_t] =&\; \sum_{\{u, v\} \subseteq E \cap (V_t \times V_t)} c^*_{uv} E[M^t_{uv}]
		= \sum_{\{u, v\} \subseteq E \cap (V_t \times V_t)} c^*_{uv} \frac{1}{|V_t|} \sum_{w \in V_t} E[M^t_{uv}|A^t_w]
	\end{align*}
	
	Our goal is to compute an upper bound on $ E[M^t_{uv}|A^t_w] $. To do so, we define several events.
	
	For each edge $ \{u, v\} $ s.t. $ c^*_{uv} = 1 $, define the following subsets of $ V_t $: $ \{u, v\} $, $ \forall i \in \{1, 2, 3\} $, $ T^{uv}_i $ is the set of vertices $ w $ s.t. $ \{u, v, w\} $ is a $ (+, +, -) $ triangle in which $ OPT $ makes exactly $ i $ mistakes, $ S^{uv} $ is the set of vertices $ w $ s.t. $ \{u, v, w\} $ is a $ (+, -, -) $ or $ (+, +, +) $ triangle in which $ OPT $ makes exactly $ 2 $ mistakes, $ R^{uv} \equiv V_t \setminus T^{uv}_1 \setminus T^{uv}_2 \setminus S^{uv} \setminus \{u, v\} $. Furthermore, let $ T^{uv}_{2u} $ be the subset of $ T^{uv}_2 $ s.t. $ w \in T^{uv}_{2u} $ if the $ 2 $ mistakes in $ \{u, v, w\} $ are both incident on $ u $. Similarly, let $ S^{uv}_u $ be the subset of $ S^{uv} $ s.t. $ w \in T^{uv}_u $ if the 2 mistakes in $ \{u, v, w\} $ are both incident on $ u $. $ \forall w \in T^{uv}_1 $, the probability that the algorithm makes a mistake on $ \{u, v\} $ given that $ w $ is the pivot is $ \Pr[D^t_{uv} | A^t_w] = 1 $. Note that $ T^{uv}_1 $, $ T^{uv}_2 $, $ \{u, v\} $, $ S^{uv} $, and $ R^{uv} $ partition $ V_t $.

	We compute $ E[M^t_{uv}|A^t_w] $ (or an upper bound thereof) when $ w $ is in each of the sets $ \{u, v\} $, $ T^{uv}_1 $, $ T^{uv}_{2u} $, $ T^{uv}_{2v} $, $ S^{uv}_u $, $ S^{uv}_v $, and $ R^{uv} $, which partition $ V_t $. Similarly, we analyze $ \Pr[D^t_{uv}|A^t_w] $, breaking up the calculation based on whether the pivot $ w $ is in $ T^{uv}_1 $, $ T^{uv}_{2u} $, $ T^{uv}_{2v} $, $ S^{uv}_{2u} $, $ S^{uv}_{2v} $, $ \{u, v\} $, or $ R^{uv} $.

	In order to prove the claim, we show that
	\begin{align*}
		\sum_{w \in V_t} E[M^t_{uv}|A^t_w] \leq&\; \max\left(2, \frac{3}{1+2p}\right) \sum_{w \in V_t}\Pr[D^t_{uv}|A^t_w]
	\end{align*}

	Thus, we have shown that $ E[ALG^t|V_t] \leq \max\left(2, \frac{3}{1+2p}\right) E[OPT^t|V_t] $. By Lemma ~\ref{lemma:loop}, the claim follows.
\end{proofsketch}
\begin{lem}
	\label{lem:randomquerypivot_queries}
	The expected number of queries made by \textsc{RandomQueryPivot} is at most $ \max\left(4p, 1\right)C_{OPT} $.
\end{lem}
\begin{proof}
	We follow an approach similar to that taken in the proof of Lemma ~\ref{lem:randomquery_mistake}. We will bound the number of queries made by the algorithm in each iteration $ t $ by charging queries to edges on which $ OPT $ makes a mistake and on which the algorithm makes a mistake in iteration $ t $. Let $ U^t $ be the number of queries made by the algorithm in iteration $ t $. We charge queries as follows to an edge $ \{u, v\} $ on which $ OPT $ makes a mistake:
	\begin{enumerate}
		\item When $ u $ or $ v $ is the pivot, the algorithm makes at most $ 1 $ query on $ \{u, v\} $ itself.
		\item When $ u $ or $ v $ is the pivot (suppose WLOG $ u $ is the pivot), $ \forall w \in T^{uv}_1 $ (defined in the proof of Lemma ~\ref{lem:randomquery_mistake}), the algorithm makes a query on $ \{u, w\} $ with probability $ p $ if $ \{u, w\} $ is a $ + $ edge.
		\item When the pivot $ w $ is in $ T^{uv}_1 $, then with probability $ p $ at most $ 2 $ queries are made when the algorithm considers the triangle $ \{u, v, w\} $.
		\item Note that we need not worry about charging mistakes in $ (+, +, -) $ triangles in which $ OPT $ makes $ 2 $ mistakes because when considering such a triangle the algorithm is guaranteed not to query the $ - $ edge on which $ OPT $ does not make a mistake. We also need not worry about charging mistakes in $ (+, +, -) $ triangles in which $ OPT $ makes $ 3 $ mistakes because each edge can be charged for any query made on that edge.
	\end{enumerate}
\begin{align*}
	E[U^t|V_t]
	\leq&\; \sum_{\{u, v\} \subseteq E \cap (V_t \times V_t)} \frac{c^*_{uv}}{|V_t|} \left[2(1+p|T^{uv}_1|) + \sum_{w \in T^{uv}_1} 2p\right]
	\leq \sum_{\{u, v\} \subseteq E \cap (V_t \times V_t)} \frac{c^*_{uv}}{|V_t|} \left[2+4p|T^{uv}_1|\right]
	\end{align*}

	Recall from the proof of Lemma ~\ref{lem:randomquery_mistake} that
	\begin{align*}
	E[OPT^t|V_t] =&\; \sum_{\{u, v\} \subseteq E \cap (V_t \times V_t)} \frac{c^*_{uv}}{|V_t|} \sum_{w \in V_t} \Pr[D^t_{uv}|A^t_w]
	\geq \sum_{\{u, v\} \subseteq E \cap (V_t \times V_t)} \frac{c^*_{uv}}{|V_t|} \left(2+|T^{uv}_1|\right)
	\end{align*}
	Here the second inequality follows from computing  $\sum_{w \in V_t} \Pr[D^t_{uv}|A^t_w]$ (see Case $1$ and $7$).
	Clearly, $ \frac{2+4p|T^{uv}_1|}{2+|T^{uv}_1|} \leq \max\left(4p, 1\right) $, so $ E[U^t|V_t] \leq \max\left(4p, 1\right) E[OPT^t|V_t] $. Then by Lemma ~\ref{lemma:loop}, the claim follows.
\end{proof}
Theorem ~\ref{thm:randomquerypivot} follows directly from Lemmas ~\ref{lem:randomquery_mistake} and ~\ref{lem:randomquerypivot_queries}.

\section{Lower bound on Query Complexity}
\label{sec:lower}
The query complexities of the algorithms presented in this paper are linear in $ C_{OPT} $, but it is not clear whether this number of queries is necessary for finding an (approximately) optimal solution. In this section, we show that a query complexity linear in $ C_{OPT} $ is necessary for approximation factors below a certain threshold assuming that the Gap-ETH, stated below, is true.
\begin{hyp}
	\label{hyp:gapeth}
	(Gap-ETH) There is some absolute constant $ \gamma > 0 $ s.t. any algorithm that can distinguish between the following two cases for any given 3-SAT instance with $ n $ variables and $ m $ clauses must take time at least $ 2^{\Omega(m)} $. (see e.g. \cite{Dinur-16})
	\begin{enumerate}[i]
		\item The instance is satisfiable.
		\item Fewer than $ (1-\gamma)m $ of the clauses are satisfiable.
	\end{enumerate}
\end{hyp}

The proof of the following lemma is provided in the appendix.
\begin{lem}
    \label{thm:lowerbound}
	Let $ C_{OPT} $ be the optimum number of mistakes for a given instance of correlation clustering. 
	Assuming Hypothesis 1, there is no $ \left(1+\frac{\gamma}{10}\right) $-approximation algorithm for correlation clustering on $ N $ vertices that runs in time $ 2^{o(C_{OPT})}\mbox{poly}(N) $ where $ \gamma $ is as defined in Hypothesis 1.
\end{lem}

As a corollary to the above lemma, we obtain the following.
\begin{thm}
\label{thm:lower1}
	There is no polynomial-time $ \left(1+\frac{\gamma}{10}\right) $-approximation algorithm for correlation clustering that uses $ o(C_{OPT}) $ queries.
\end{thm}
\begin{proof}
Suppose there exists an algorithm that approximates correlation clustering with an approximation factor of $ \left(1+\frac{\gamma}{10}\right) $ and uses at most $ o(C_{OPT}) $ queries. We follow the algorithm but instead when the algorithm issues a query, we branch to two parallel solutions instances with the two possible query answers from the oracle. Since the number of queries is $o(C_{OPT})$, the number of branches/solutions that we obtain by this process is at most $2^{o(C_{OPT})}$. We return the one which gives the minimum number of mistakes. This gives a contradiction to Lemma~\ref{thm:lowerbound}.
\end{proof}

\section{Experiments}
\label{sec:experiment}
In this section, we report detailed experimental results on multiple synthetic and real-world datasets. We compare the performance of the existing correlation clustering algorithms that do not issue any queries, alongside with our new algorithms. We compare three existing algorithms:  the deterministic constant factor approximation algorithm of Bansal et al. \cite{Bansal-02} (BBC), the combinatorial 3-approximation algorithm of  Ailon et al. \cite{Ailon-05} (ACN), and the state-of-the-art $2.06$-approximation algorithm of Chawla et al. based on linear program (LP) rounding \cite{Chawla-15} (LP-Rounding). The code and data used in our experiments can be found at \texttt{\href{https://github.com/sanjayss34/corr-clust-query-esa2019}{https://github.com/sanjayss34/corr-clust-query-esa2019}}.

\paragraph*{Datasets.} Our datasets range from small synthetic datasets to large real datasets and real crowd answers obtained using Amazon Mechanical Turk. Below we give a short description of them. 

\noindent {\bf Synthetics Datasets: Small.} We generate graphs with $\approx 100$ nodes by varying the cluster size distribution as follows.  [N] represents $ 10$ cliques whose sizes are drawn i.i.d. from a $ Normal(8, 2) $ distribution. This generates clusters of nearly equal size.
[S] represents $5$ clusters of size $5$ each, $4$ clusters of size $15$ each and one cluster of size $30$. This generates clusters with moderate skew. [D] represents $ 3 $ cliques whose total size is $ 100 $ and whose individual sizes are determined by a draw from a $ Dirichlet((3, 1, 1)) $ distribution. This generates clusters with extreme skewed distribution with one cluster accounting for more than $80\%$ of edges.

\noindent{\bf Synthetics Datasets: Large.}
We generate two datasets {\em skew} and {\em sqrtn} each containing $900$ nodes of fictitious hospital patients data, including name, phone number, birth date and address using the data set generator of the Febrl system \cite{christen2008febrl}. {\em skew} contains few $(\approx \log{n})$ clusters of large size $(\approx \frac{n}{\log{n}})$, moderate number of clusters $(\approx \sqrt{n})$ of moderate size $(\approx \sqrt{n})$ and a large tail of small clusters. {\em sqrtn} contains $\sqrt{n}$ clusters of size $\sqrt{n}$.

\noindent{ \bf Noise Models for Synthetic Datasets.}
Initially, all intra-cluster edges are labelled with $+$ sign and all inter-cluster edges are labelled with $-$ sign. Next, the signs of a subset of edges are flipped according to  the following distributions. Denote by $ C_1, C_2, ..., C_k $ the clusters that we generate. Let $N$ denote the number of vertices in a graph. Let $ \ell_1 = 0.01 $, $ \ell_2 = 0.1 $, and $ L $ be an integer. For the small datasets, we set $ L = 100 $, and for the large {\em skew} and {\em sqrtn} datasets, we set $ L = \lfloor \ell_2\binom{N}{2} \rfloor $.
\begin{itemize}
    \item [I.] Flip sign of $ L $ edges uniformly at random.
    \item [II.] Flip sign of $ \min\{\lfloor L/k\rfloor, |C_i|-1\} $ edges uniformly at random within each clique $ C_i $. Do not flip sign of the inter-cluster edges.
    \item [III.] Flip sign of edges as in II in addition to selecting uniformly at random $ \lceil \ell_1 |C_i||C_j| \rceil $ edges between each pair of cliques $ C_i $, $ C_j $ and flipping their sign.
\end{itemize}

\noindent{\bf Real-World Datasets.} We use several real-world datasets.

$\bullet$ In the {\em cora} dataset \cite{cora}, each node is a scientific paper represented by a string determined by its title, authors, venue, and date; edge weights between nodes are computed using Jaro string similarity \cite{firmani2018robust, winkler2006overview}. The {\em cora} dataset consists of $1.9K$ nodes, $191$ clusters with the largest cluster-size being $236$.

$\bullet$  In the {\em gym} dataset \cite{verroios2015entity}, each node corresponds to an image of a gymnast, and each edge weight reflects the similarity of the two images (i.e. whether the two images correspond to the same person). The {\em gym} dataset consists of $94$ nodes with $12$ clusters and maximum cluster size is $15$.
    
    $\bullet$ In the {\em landmarks} dataset \cite{gruenheid2015fault}, each node corresponds to an image of a landmark in Paris or Barcelona, and the edge weights reflect the similarity of the two images. The {\em landmarks} dataset consists of $266$ nodes, $13$ clusters and the maximum size of clusters is $43$.
    
    $\bullet$ In the {\em allsports} dataset \cite{verroios2017waldo}, the nodes correspond to images of athletes in one of several sports, and the edge weights reflect the similarity of the two images.  The pairs of images across sports are easy to distinguish but the images within the same category of sport are quite difficult to distinguish due to various angles of the body, face and uniform. The {\em allsports} dataset consists of $200$ nodes with $64$ clusters and with a maximum size of cluster being just $5$.

Since the underlying graphs are weighted, we convert the edge weights to $\pm 1$ labels by simply labeling an edge $ + $ if its weight is at least $ 1/2 $ and $ - $ otherwise (the edge weights in all of the weighted graphs are in $ [0, 1] $). We also perform experiments directly on the weighted graphs \cite{Chawla-15} to show how the above rounding affects the results.

\noindent{\bf Oracle.} For small datasets, we use the Gurobi (\href{www.gurobi.org}{www.gurobi.org}) optimizer to solve the integer linear program (ILP) for correlation clustering \cite{Chawla-15} to obtain the optimum solution, which is then used as an oracle. For larger datasets like {\em skew}, {\em sqrtn} and {\em cora}, ILP takes prohibitively long time to run. For these large datasets,the ground-truth clustering is available and is used as the oracle. 

For practical implementation of oracles, one can use the available crowd-sourcing platforms such as the Amazon Mechanical Turk. It is possible that such an oracle may not always give correct answer. We also use such crowd-sourced oracle for experiments on real datasets. Each question is asked $3$ to $5$ times to Amazon Mechanical Turk, and a majority vote is taken to resolve any conflict among the answers. We emphasize that the same-cluster query setting can be useful in practice because two different sources of information can produce the edge signs and the oracle -- for instance, the edge signs can be produced by a cheap, automated computational method (e.g. classifiers), while the oracle answers can be provided by humans through the crowd-sourcing mechanism explained above.
\begin{table}[pt]
    \centering
    \begin{tabular}{|c|c|c|c|c|c|c|c|c|c|c|}
    \hline
      Mode  & \shortstack{ILP \\ Oracle} & BBC & ACN & \shortstack{LP \\ Rounding} & Bocker & \shortstack{Bocker \\ Queries} & QP & \shortstack{QP \\ Queries} & RQP & \shortstack{RQP \\ Queries}  \\
      \hline
      N+I & 100 & 271 & 205.67 & 100 & 100 & 40 & 100 & 113 & 104.33 & 63.3 \\
      \hline
      N+II & 48 & 104 & 70.0 & 48 & 48 & 30 & 48 & 47 & 56.33 & 27.33 \\
      \hline
      N+III & 93 & 201 & 130 & 123 & 93 & 49 & 93 & 91 & 97.67 & 66.0 \\
      \hline
      D+I & 100 & 100 & 267.0 & 100 & 100 & 46 & 100 & 86 & 100 & 83.33 \\
      \hline
      D+II & 48 & 48 & 144.33 & 48 & 48 & 34 & 48 & 57 & 48 & 40.67 \\
      \hline
      D+III & 64 & 64 & 216.33 & 64 & 64 & 43 & 64 & 71 & 64 & 75.0 \\
      \hline
      S+I & 100 & 969 & 206.0 & 100 & 100 & 47 & 100 & 136 & 100.67 & 92.0 \\
      \hline
      S+II & 60 & 831 & 100.67 & 61.67 & 60 & 46 & 60 & 71 & 63.67 & 58.67 \\
      \hline
      S+III & 137 & 913 & 297 & 141.33 & 137 & 71 & 137 & 159 & 139.33 & 107.33 \\
      \hline
    \end{tabular}
    \caption{Results for Experiments on synthetic small datasets. BBC denotes the algorithm of \cite{Bansal-02}, ACN denotes the 3-approximation algorithm of \cite{Ailon-05}, LP Rounding denotes the algorithm of \cite{Chawla-15}, QP denotes \textsc{QueryPivot}, and RQP denotes \textsc{RandomQueryPivot}(0.25). All numerical columns except those marked as ``Queries'' give the number of mistakes made by the algorithm.}
    \label{tab:basic_experiments}
    \vspace{-1em}
\end{table}
\vspace{-1em}
\paragraph*{Results.} We compare the results of our  \textsc{QueryPivot} and \textsc{RandomQueryPivot} algorithm as well as the prior algorithms BBC \cite{Bansal-02}, ACN \cite{Ailon-05}, LP-Rounding \cite{Chawla-15}, and one of Bocker's edge branching algorithms \cite{Bocker2009}. For the algorithms that are randomized (ACN, LP-Rounding and \textsc{RandomQueryPivot}), we report the average of three runs. The algorithm of Bansal et al. \cite{Bansal-02} requires setting a parameter $\delta$. We tried several values of $ \delta $ on several of the datasets and chose the value that seemed to give the best performance overall.

\noindent{\bf Synthetic Datasets.}
Table \ref{tab:basic_experiments} summarizes the results of different algorithms on small synthetic datasets. 
\begin{table}[pt]
    \centering
   \small \begin{tabular}{|c|c|c|c|c|c|c|c|c|c|}
    \hline
    Dataset/Mode & \shortstack{LP \\ Rounding} & BBC & ACN & Bocker & \shortstack{Bocker \\ Queries} & QP & QP Queries & RQP & RQP Queries  \\
      \hline
      Skew (I) & & 8175 & 31197.67 & 48 & 1054 & 0 & 17108 & 71.33 & 10051.33 \\
      \hline
      Skew (II) & & 700 & 1182.33 & 416 & 377 & 60 & 668 & 282 & 558.0 \\
      \hline
      Skew (III) & & 8175 & 12260.67 & 379 & 1370 & 56 & 8977 & 293.0 & 4475.67 \\
      \hline
      Sqrtn (I) & & 13050 & 36251.33 & 0 & 862 & 0 & 13171 & 9.67 & 7851.0 \\
      \hline
      Sqrtn (II) & & 0 & 1484.67 & 0 & 494 & 0 & 748 & 0.0 & 711.0 \\
      \hline
      Sqrtn (III) & & 13050 & 12711.33 & 0 & 841 & 0 & 6449 & 0.0 & 2693.33 \\
      \hline
    \end{tabular}
    \caption{Results for Large Synthetic Datasets where Mistakes are measured with respect to ground-truth clustering and the oracle is the ground-truth clustering.}
    \label{tab:large_datasets-1}
    \vspace{-1em}
\end{table}

As we observe, our \textsc{QueryPivot} algorithm always obtains the optimum clustering. Moreover, \textsc{RandomQueryPivot} has a performance very close to \textsc{QueryPivot} but often requires much less queries. Interestingly, the LP-rounding algorithm performs very well except for $N+III$. ACN and BBC algorithms have worse performance than LP-Rounding, and in most cases ACN is preferred over BBC. The Bocker algorithm obtains the optimal clustering as well and, with the exception of one case, uses fewer queries than \textsc{RandomQueryPivot}.
 
For the larger synthetic datasets {\em skew} and {\em sqrt}, as discussed the ground-truth clustering is used as an oracle. We also use the ground-truth clustering to count the number of mistakes. On these datasets, the LP-rounding algorithm caused an out-of-memory error on a machine with 256 GB main memory that we used. The linear programming formulation for correlation clustering has $O(n^3)$ triangle inequality constraints; this results in very high time and space complexity rendering the LP-rounding impractical for correlation clustering on large datasets. Table \ref{tab:large_datasets-1} summarizes the results.

As we observe, \textsc{QueryPivot} algorithm recovers the exact ground-truth clustering in several cases. \textsc{RandomQueryPivot} has a low error rate as well and uses significantly fewer queries. Compared to the Bocker algorithm, \textsc{QueryPivot} and \textsc{RandomQueryPivot} generally make fewer mistakes but use more queries.

\noindent{\bf Real-World Datasets}
The results for the real-world datasets are reported in Table \ref{tab:large_datasets-2}, \ref{tab:large_datasets-3} and \ref{tab:large_datasets-4}. It is evident from Table~\ref{tab:large_datasets-2} that our algorithms outperform the existing algorithms aside from the Bocker algorithm by a big margin in recovering the original clusters. Table \ref{tab:large_datasets-2} also includes results for the LP-rounding algorithm applied to the original weighted graph for the Gym, Landmarks, and Allsports datasets. We also report in Table~\ref{tab:large_datasets-2time} the running times for the experiments in Table~\ref{tab:large_datasets-2}. These numbers show that the BBC and ACN algorithms are substantially faster than the others, while our algorithms are substantially faster than the LP-rounding algorithm. The Bocker algorithm is considerably slower than our algorithms on both the Landmarks dataset and the Cora dataset, which is the largest. We note that of the three ``data reduction'' techniques described in \cite{Bocker2009}, we implemented two -- removing cliques in intermediate ``edge branching (querying)'' steps and merging vertices according to queries. The technique that we did not implement, ``checking for unaffordable edge modifications'' assumes that the number of mistakes made by the optimal clustering is known.
\begin{table}[pt]
    \centering
   \small \begin{tabular}{|c|c|c|c|c|c|c|c|c|c|c|}
    \hline
    Dataset/Mode & \shortstack{LP \\ Rounding} & \shortstack{LP \\ Rounding \\ (weighted)} & BBC & ACN & Bocker & \shortstack{Bocker \\ Queries} & QP & \shortstack{QP \\ Queries} & RQP & \shortstack{RQP \\ Queries}  \\
      \hline
      Cora &  &  & 62891 & 26065.0 & 8164 & 2004 & 4526 & 2188 & 4664.67 & 1474.33 \\
      \hline
      Gym & 221.0 & 332.67 & 449 & 301.67 & 65 & 74 & 8 & 150 & 82.67 & 97.33 \\
      \hline
      Landmarks & 29648.0 & 25790.0 & 31507 & 28770 & 238 & 1593 & 3426 & 953 & 1124.67 & 1467.0 \\
      \hline
      Allsports & 230.0 & 226.33 & 227 & 253.33 & 223 & 13 & 217 & 41 & 223.67 & 21.0 \\
      \hline
    \end{tabular}
    \caption{Results for Real-World Datasets where mistakes are measured with respect to ground-truth clustering and the oracle is the ground-truth clustering. LP Rounding (weighted) refers to the LP rounding of \cite{Chawla-15} applied to the weighted input graph.}
    \label{tab:large_datasets-2}
    \vspace{-1em}
\end{table}
\begin{table}[pt]
    \centering
   \begin{tabular}{|c|c|c|c|c|c|c|c|}
    \hline
    Dataset/Mode & \shortstack{LP \\ Rounding} & \shortstack{LP \\ Rounding \\ (weighted)} & BBC & ACN & Bocker & QP & RQP \\
      \hline
      Cora & & & 1.58 & 0.16 & 6182.53 & 2170.33 & 515.59 \\
      \hline
      Gym & 4.96 & 5.09 & 0.004 & 0.0014 & 0.42 & 0.33 & 0.27 \\
      \hline
      Landmarks & 190.96 & 9571.32 & 0.048 & 0.00067 & 65.28 & 1.96 & 2.82 \\
      \hline
      Allsports & 41.54 & 42.08 & 0.018 & 0.025 & 0.32 & 13.28 & 12.76 \\
      \hline
    \end{tabular}
    \caption{Running times (in seconds) for the results in Table \ref{tab:large_datasets-2}. For randomized algorithms, the time shown is the average over three trials.\protect\footnotemark}
    \label{tab:large_datasets-2time}
    \vspace{-1em}
\end{table}
\footnotetext{In these experiments, we used a machine running Ubuntu 16 with 28 2.6 GHz Intel Xeon E5-2690 v4 CPU's and 256 GB of main memory.}

\begin{table}[H]
\centering
\begin{tabular}{|c|c|c|c|c|c|c|}
    \hline
    Dataset/Mode & Bocker & \shortstack{Bocker \\ Queries} & QP & QP Queries & RQP & RQP Queries  \\
      \hline
      Gym & 156 & 81 & 135 & 175 & 160.0 & 104.67 \\
      \hline
      Landmarks & 1221 & 3139 & 4645 & 1997 & 2172.33 & 1548.33 \\
      \hline
      Allsports & 223 & 13 & 218 & 41 & 223.67 & 21.0 \\
      \hline
    \end{tabular}
     \caption{Results for Real-World Datasets where mistakes are measured with respect to ground-truth clustering and the oracle is the crowd.}
    \label{tab:large_datasets-3}
   \vspace{-1em}
\end{table}
Table~\ref{tab:large_datasets-3} reports the results using a faulty crowd oracle. Contrasting the results of Table~\ref{tab:large_datasets-2} and \ref{tab:large_datasets-3}, we observe minimal performance degradation; that is, our algorithms are robust to noise. The results in this table are important, as this setting is closest to the typical real-world application of same-cluster queries. Note that the source of information that gives the signs of the edges is different from that which is the crowd oracle. For the landmarks dataset, the original edge weights are determined by a gist detector \cite{oliva2001modeling}, while the oracle used in Table ~\ref{tab:large_datasets-3} is given by high-quality crowd workers. For the gym and allsports datasets, the original edge weights are determined by (lower quality) human crowd workers, but the oracle used in Table ~\ref{tab:large_datasets-3} is based on high-quality crowd workers.
Finally, in Table~\ref{tab:large_datasets-4}, we report the results using the optimum ILP solution as the oracle. For the larger datasets, it is neither possible to run the ILP nor LP-Rounding due to their huge space and time requirements. In general, our algorithms \textsc{QueryPivot} and \textsc{RandomQueryPivot} outperform the other algorithms except for Bocker et al.'s \cite{Bocker2009} algorithm. In terms of number of mistakes and query complexity, our algorithms are comparable to Bocker et al.'s algorithm; there are cases in which the latter attains superior performance and cases in which our algorithms are better. We also note that our algorithms are in general faster than Bocker et al.'s algorithm. A more detailed analysis of the comparison among our algorithms and Bocker et al.'s algorithm is left as a topic for future work.

 \begin{table}[H]
    \centering
    \begin{tabular}{|c|c|c|c|c|c|c|c|c|c|}
    \hline
    Dataset/Mode &  \shortstack{LP \\ Rounding} & BBC & ACN & Bocker & \shortstack{Bocker \\ Queries} & QP & QP Queries & RQP & RQP Queries  \\
      \hline
      Gym & 276.0 & 464 & 338.0 & 207 & 80 & 207 & 171 & 211.0 & 112.67 \\
      \hline
      Landmarks & 4092.0 & 4995 & 5240.67 & 4092 & 254 & 4092 & 267 & 4092.0 & 265.33 \\
      \hline
      Allsports & 33.33 & 65 & 40.67 & 28 & 12 & 28 & 36 & 30.33 & 18.67 \\
      \hline
    \end{tabular}
    \caption{Results for Real-World Datasets where mistakes are measured with respect to the graph and the oracle is the optimal ILP solution for the graph.}
    \label{tab:large_datasets-4}
    \vspace{-1em}
\end{table}

\section{Acknowledgements}
The second author would like to thank Dan Roth for letting him use his machines for running experiments, Sainyam Galhotra for help with datasets, and Rajiv Gandhi for useful discussions.

%%
%% Bibliography
%%

%% Please use bibtex, 

%\bibliographystyle{unsrtnat}
\newpage
\bibliography{lipics-v2019-sample-article}

% \bibliography{example-paper}
%%%%%%%%%%%%%%%%%%%%%%%%%%%%%%%%%%%%%%%%%%%%%%%%%%%%%%%%%%%%%%%%%%%%%%%%%%%%%%%
%%%%%%%%%%%%%%%%%%%%%%%%%%%%%%%%%%%%%%%%%%%%%%%%%%%%%%%%%%%%%%%%%%%%%%%%%%%%%%%

\cleardoublepage

\section*{Appendix}
\subsection{Proof of Lemma~\ref{lem:randomquery_mistake}}

Let $ s_{uv} = 1 $ if $ \{u, v\} $ is a $ - $ edge and $ 0 $ otherwise. Let $ c^*_{uv} $ equal $ 1 $ if $ OPT $ makes a mistake on $ \{u, v\} $ and $ 0 $ otherwise. 

Let $ OPT^t $ be the number of edges $ \{u, v\} $ s.t. $ c^*_{uv} = 1 $ and the algorithm makes a decision on $ \{u, v\} $ in iteration $ t $. Let $ ALG^t $ be the number of edges $ \{u, v\} $ s.t. the algorithm makes a mistake on $ \{u, v\} $ and the algorithm makes a decision on $ \{u, v\} $ in iteration $ t $. 

Let $ V_t $ be the set of vertices remaining at the beginning of iteration $ t $. Let $ D^t_{uv} $ be the event that the algorithm makes a decision on $ \{u, v\} $ in iteration $ t $.

{\bf Lemma~\ref{lem:randomquery_mistake}.}
	{\it The expected number of mistakes made by the algorithm's clustering is at most $ \max\left(2, \frac{3}{1+2p}\right)C_{OPT} $.}
\begin{proof}
	By Lemma ~\ref{lemma:loop}, if we show that $ E[ALG^t - \alpha OPT^t] \leq 0 $ for any $ t $, where $ \alpha \leq \max\left(2, \frac{3}{1+2p}\right)$, then the claim will follow.  Let $ A^t_w $ be the event that $ w \in V_t $ is the pivot in iteration $ t $. Since a pivot is selected uniformly at random from $V_t$, we have
	\begin{align*}
		E[OPT^t | V_t] =&\; \sum_{\{u, v\} \subseteq E \cap (V_t \times V_t)} \frac{c^*_{uv}}{|V_t|} \sum_{w \in V_t} \Pr[D^t_{uv}|A^t_w]
	\end{align*}
	
	Now we will write $ E[ALG^t|V_t] $ by charging the algorithm's mistakes to each of $ OPT $'s mistakes. Let $ M^t_{uv} $ be the charge incurred to $ \{u, v\} $ in iteration $ t $. We will assign charges such that $M^t_{uv}=0$ if $c^*_{uv}=0$. Then
	\begin{align*}
		E[ALG^t|V_t] =&\; \sum_{\{u, v\} \subseteq E \cap (V_t \times V_t)} c^*_{uv} E[M^t_{uv}] \\
		=&\; \sum_{\{u, v\} \subseteq E \cap (V_t \times V_t)} c^*_{uv} \frac{1}{|V_t|} \sum_{w \in V_t} E[M^t_{uv}|A^t_w]
	\end{align*}
	
	Our goal is to compute an upper bound on $ E[M^t_{uv}|A^t_w] $ when $ c^*_{uv} = 1 $. To do so, we need to define several events.
	
	We now consider a fixed edge $ \{u, v\} $ s.t. $ c^*_{uv} = 1 $, define the following subsets of $ V_t $: 
	\begin{itemize}
	\item $ \forall i \in \{1, 2, 3\} $, $ T^{uv}_i $ is the set of vertices $ w $ s.t. $ \{u, v, w\} $ is a $ (+, +, -) $ triangle in which $ OPT $ makes exactly $ i $ mistakes, 
	\item $ S_2^{uv} $ is the set of vertices $ w $ s.t. $ \{u, v, w\} $ is a $ (+, -, -) $ or $ (+, +, +) $ triangle in which $ OPT $ makes exactly $ 2 $ mistakes, 
	\item $ Y^{uv}_u $ is $ \emptyset $ if $ \{u, v\} $ is a $ + $ edge; if $ \{u, v\} $ is a $ - $ edge, then $ Y^{uv}_u $ is the set of vertices $ w $ s.t. $ \{u, w\} $ is a $ + $ edge, $ \{v, w\} $ is a $ - $ edge, and $ OPT $ makes mistakes on $ \{u, v\} $ and $ \{u, w\} $. Since $ \{u, v, w\} $ is a $ (+, -, -) $ triangle in which $ OPT $ makes exactly two mistakes, $ Y^{uv}_u \subseteq S_2^{uv} $.
	\item $ R^{uv} \equiv V_t \setminus \{T^{uv}_1 \cup T^{uv}_2 \cup S_2^{uv} \cup \{u, v\} \}$. 
		\item Furthermore, let $ T^{uv}_{2u} $ be the subset of $ T^{uv}_2 $ s.t. $ w \in T^{uv}_{2u} $ if the $ 2 $ mistakes of $ OPT $ in $ \{u, v, w\} $ are both incident on $ u $. 
	\item Similarly, let $ S^{uv}_{2u} $ be the subset of $ S_2^{uv} $ s.t. $ w \in S^{uv}_{2u} $ if the 2 mistakes of $ OPT $ in $ \{u, v, w\} $ are both incident on $ u $.
	\end{itemize}
	 Note that $ T^{uv}_1 $, $ T^{uv}_2 $, $ \{u, v\} $, $ S_2^{uv} $, and $ R^{uv} $ partition $ V_t $.
	
	Let $ Q^t_{uv} $ be the event that $ \{u, v\} $ is queried in iteration $ t $. Note that by the argument in the proof of Lemma ~\ref{lemma:queryequalprob}, $ \Pr[Q^t_{uv} | A^t_u] = \Pr[Q^t_{uv} | A^t_v] = 1-(1-p)^{|T^{uv}_1|+|T^{uv}_2|+|T^{uv}_3|} $ if $ \{u, v\} $ is a $ + $ edge and that $ \Pr[Q^t_{uv} | A^t_u] = \Pr[Q^t_{uv} | A^t_v] = 1-(1-p)^{|T^{uv}_1|} $ if $ \{u, v\} $ is a $ - $ edge. Also, note that in any $ (+, +, -) $ triangle in which $ OPT $ makes exactly $ 2 $ mistakes, both mistakes must be on $ + $ edges.
	
	Let $ M^t_{uv} $ be the charge incurred to $ \{u, v\} $ in iteration $ t $. Then
	\begin{align*}
		E[ALG^t|V_t] =&\; \sum_{\{u, v\} \subseteq E \cap (V_t \times V_t)} c^*_{uv} E[M^t_{uv}] \\
		=&\; \sum_{\{u, v\} \subseteq E \cap (V_t \times V_t)} c^*_{uv} \frac{1}{|V_t|} \sum_{w \in V_t} E[M^t_{uv}|A^t_w]
	\end{align*}
	We now compute $ E[M^t_{uv}|A^t_w] $ (or an upper bound thereof) when $ w $ is in each of the sets $ \{u, v\} $, $ T^{uv}_1 $, $ T^{uv}_{2u} $, $ T^{uv}_{2v} $, $ S^{uv}_{2u} $, $ S^{uv}_{2v} $, and $ R^{uv} $, which partition $ V_t $.
	\begin{enumerate}
		\item $ w \in T^{uv}_1 $. Then $ \{u, v\} $ is in exactly one triangle that includes the pivot $ w $, and the algorithm will make a mistake on $ \{u, v\} $ in iteration $ t $ regardless of whether $ uw $ or $ vw $ is queried. Therefore, $ E[M^t_{uv}|A^t_w] = \Pr[D^t_{uv} | A^t_w] =1 $. 
		\item $ w \in T^{uv}_{2u} $. Then $ \{u, v\} $ is in exactly one triangle that includes the pivot $ w $, and the algorithm will make a mistake on $ \{u, v\} $ iff the algorithm does not query $ \{u, w\} $. Therefore, $ E[M^t_{uv}|A^t_w] = 1-\Pr[Q^t_{uw}|A^t_w] $.
		\item $ w \in T^{uv}_{2v} $. Analogous to case 2: $ E[M^t_{uv}|A^t_w] = 1-\Pr[Q^t_{vw}|A^t_w] $.
		\item $ w \in Y^{uv}_u $. Then recall that $ \{u, v\} $ is a $ - $ edge, $ \{u, w\} $ is a $ + $ edge, and $ \{v, w\} $ is a $ - $ edge. In this case, whether or not the algorithm queries $ \{u, w\} $, the algorithm will not make a mistake on $ \{u, v\} $, so $ E[M^t_{uv}|A^t_w] = 0 $.
		\item $ w \in Y^{uv}_v $. Analogous to case 4: $ E[M^t_{uv}|A^t_w] = 0 $.
		\item $ w \in S^{uv}_{2u} \setminus Y^{uv}_u $. Then $ \{u, v\} $ is in exactly one triangle that includes the pivot $ w $, and the algorithm will make a mistake on $ \{u, v\} $ iff the algorithm queries $ \{u, w\} $. Therefore, $ E[M^t_{uv}|A^t_w] = \Pr[Q^t_{uw}|A^t_w] $.
		\item $ w \in S^{uv}_{2v} \setminus Y^{uv}_v $. Analogous to case 6: $ E[M^t_{uv}|A^t_w] = \Pr[Q^t_{vw}|A^t_w] $.
		\item $ w \in R^{uv} $. Then $ \{u, v\} $ is in exactly one triangle that includes the pivot $ w $, and clearly $ \Pr[M^t_{uv}|A^t_w] \leq \Pr[D^t_{uv}|A^t_w] $.
		\item $ w = u $. The expected charge in this case is equal to the sum of the following parts.
		\begin{enumerate}
			\item The algorithm makes a mistake on $ \{u, v\} $ iff the algorithm queries $ \{u, v\} $. The charge for this part is thus $ \Pr[Q^t_{uv}|A^t_u] $.
			
			\item If the algorithm does not query $ \{u, v\} $, then $ \forall w \in T^{uv}_1 $, the algorithm makes a mistake on $ \{v, w\} $. The charge for this part is thus $ |T^{uv}_1|(1-\Pr[Q^t_{uv}|A^t_u]) $.
			
			\item For each $ w \in T^{uv}_{2u} $, the algorithm makes a mistake on $ \{v, w\} $ iff the algorithm queries neither $ \{u, v\} $ nor $ \{u, w\} $. In this case, we charge $ \frac{1}{2} $ to $ \{u, v\} $ and $ \frac{1}{2} $ to $ \{u, w\} $. Thus, the expected charge to $ \{u, v\} $ for this part is $ \frac{1}{2} \sum_{w \in T^{uv}_{2u}} \Pr[\overline{Q^t_{uv}} \cap \overline{Q^t_{uw}} | A^t_u] \leq \frac{1}{2} \sum_{w \in T^{uv}_{2u}} \Pr[\overline{Q^t_{uw}}|A^t_u] = \frac{1}{2} \sum_{w \in T^{uv}_{2u}} 1-\Pr[Q^t_{uw}|A^t_u] $.
			
			\item For each $ w \in S^{uv}_{2u} \setminus Y^{uv}_u $, the algorithm makes a mistake on $ \{v, w\} $ only if the algorithm queries exactly one of $ \{u, v\} $ and $ \{u, w\} $. We will charge $ \{u, v\} $ for a mistake on $ \{v, w\} $ in the case that the algorithm queries $ \{u, w\} $ and not $ \{u, v\} $ (and we will charge $ \{u, w\} $ otherwise). The expected charge for this part is then $ \sum_{w \in S^{uv}_{2u} \setminus Y^{uv}_u} \Pr[\overline{Q^t_{uv}} \cap Q^t_{uw} | A^t_u] \leq \sum_{w \in S^{uv}_{2u} \setminus Y^{uv}_u} \Pr[Q^t_{uw} | A^t_u] $.
			
			\item We now argue that for any other vertex $ w $, we need not charge anything more to $ \{u, v\} $ due to a mistake on the edge $ \{v, w\} $. If $ w \in T^{uv}_{2v} $, then $ OPT $ makes a mistake on $ \{v, w\} $, so we charge any mistake made on $ \{v, w\} $ to $ \{v, w\} $. Similarly, if $ w \in S^{uv}_{2v} $ or if $ w \in T^{uv}_3 $, then we charge any mistake made on $ \{v, w\} $ to $ \{v, w\} $. If $ w \in Y^{uv}_u $, then the algorithm will make a mistake on $ \{v, w\} $ iff it queries $ \{u, v\} $ and not $ \{u, w\} $; the charge for this mistake is assigned to $ \{u, w\} $ by case (d). If for some vertex $ w $, $ \{u, v, w\} $ is a $ (+, -, -) $ or a $ (+, +, +) $ in which $ OPT $ makes $ 0 $ mistakes, then it is easily verified that the algorithm will not make a mistake on one of the edges in iteration $ t $ given that $ u $ is the pivot. If for some vertex $ w $, $ \{u, v, w\} $ is a $ (+, -, -) $ triangle in which $ OPT $ makes exactly $ 1 $ mistake, then the algorithm will not make a mistake on $ \{v, w\} $ regardless of whether $ \{u, v\} $ is queried. (Note that it is not possible for $ OPT $ to make exactly $ 1 $ mistake in a $ (+, +, +) $ triangle. If for some vertex $ w $, $ \{u, v, w\} $ is a $ (+, +, +) $ triangle in which $ OPT $ makes $ 3 $ mistakes, then we charge any mistake made on $ \{v, w\} $ to $ \{v, w\} $, on which $ OPT $ must make a mistake. Finally, if for some vertex $ w $ $ \{u, v, w\} $ is a $ (-, -, -) $ triangle, then we make a mistake on $ \{v, w\} $ in iteration $ t $ iff we make mistakes on both $ \{u, v\} $ and $ \{u, w\} $ in iteration $ t $. Note that this event can occur only if $ OPT $ also makes a mistake on $ \{v, w\} $. Thus, in this case too we charge a mistake made on $ \{v, w\} $ to $ \{v, w\} $ and not to $ \{u, v\} $.
		\end{enumerate}
		Total (upper bound): $ \Pr[Q^t_{uv}|A^t_u] + |T^{uv}_1|(1-\Pr[Q^t_{uv}|A^t_u]) + \frac{1}{2} \sum_{w \in T^{uv}_{2u}} 1-\Pr[Q^t_{uw}|A^t_u] + \sum_{w \in S^{uv}_{2u} \setminus Y^{uv}_u} \Pr[Q^t_{uw} | A^t_u] $
		
		\item $ w = v $. Analogous to the previous case: Total (upper bound) is $ \Pr[Q^t_{uv}|A^t_v] + |T^{uv}_1|(1-\Pr[Q^t_{uv}|A^t_v]) + \frac{1}{2} \sum_{w \in T^{uv}_{2v}} 1-\Pr[Q^t_{vw}|A^t_v] + \sum_{w \in S^{uv}_{2v} \setminus Y^{uv}_v} \Pr[Q^t_{vw} | A^t_v] $
	\end{enumerate}
	
	Adding the expected charges (or upper bounds thereof) for each of these cases, we obtain:
	\begin{align*}
		\sum_{w \in V_t} E[M^t_{uv}|A^t_w] \leq&\; \Pr[Q^t_{uv}|A^t_u] + \Pr[Q^t_{uv}|A^t_v] + |T^{uv}_1| \\
		+&\; 2|T^{uv}_1|(1-\Pr[Q^t_{uv}|A^t_u]) \\
		+&\; \frac{3}{2} \sum_{w \in T^{uv}_{2u}} 1- \Pr[Q^t_{uw}|A^t_u] \\
		+&\; \frac{3}{2} \sum_{w \in T^{uv}_{2v}} 1-\Pr[Q^t_{vw}|A^t_v] \\
		+&\; 2 \sum_{w \in S^{uv}_{2u} \setminus Y^{uv}_u} \Pr[Q^t_{uw}|A^t_u] \\
		+&\; 2 \sum_{w \in S^{uv}_{2v} \setminus Y^{uv}_v} \Pr[Q^t_{vw}|A^t_v] \\
		+&\; \sum_{w \in R^{uv}} \Pr[D^t_{uv}|A^t_w]
	\end{align*}
	Note that we used Lemma ~\ref{lemma:queryequalprob} to group some terms above.
	
	Now we perform a similar analysis for $ \Pr[D^t_{uv}|A^t_w] $, breaking up the calculation based on whether the pivot $ w $ is in $ T^{uv}_1 $, $ T^{uv}_{2u} $, $ T^{uv}_{2v} $, $ S^{uv}_{2u} $, $ S^{uv}_{2v} $, $ \{u, v\} $, or $ R^{uv} $.
	\begin{enumerate}
		\item $ w \in T^{uv}_1 $. Then the algorithm is guaranteed to make a mistake on $ \{u, v\} $ in iteration $ t $, so $ \Pr[D^t_{uv}|A^t_w] = 1 $.
		
		\item $ w \in T^{uv}_{2u} $. Note that the probability that the algorithm makes a mistake on $ \{u, v\} $ in iteration $ t $ is then $ 1-\Pr[Q^t_{uw}|A^t_w] $, as noted in the case analysis for $ E[M^t_{uv}|A^t_w] $. Then since $ \Pr[D^t_{uv}|A^t_w] $ is at least the probability that the algorithm makes a mistake on $ \{u, v\} $ in iteration $ t $, $ \Pr[D^t_{uv}|A^t_w] \geq 1-\Pr[Q^t_{uw}|A^t_w] $.
		
		\item $ w \in T^{uv}_{2v} $. Analogous to case 2: $ \Pr[D^t_{uv}|A^t_w] \geq 1-\Pr[Q^t_{vw}|A^t_w] $.
		
		\item $ w \in S^{uv}_{2u} \setminus Y^{uv}_u$. Note that the probability that the algorithm makes a mistake on $ \{u, v\} $ in iteration $ t $ is equal to $ \Pr[Q^t_{uw}|A^t_w] $, as noted in the case analysis for $ E[M^t_{uv}|A^t_w] $. Then since $ \Pr[D^t_{uv}|A^t_w] $ is at least the probability that the algorithm makes a mistake on $ \{u, v\} $ in iteration $ t $, $ \Pr[D^t_{uv}|A^t_w] \geq \Pr[Q^t_{uw}|A^t_w] $.
		
		\item $ w \in S^{uv}_{2v} \setminus Y^{uv}_v $. Analogous to case 4: $ \Pr[D^t_{uv}|A^t_w] \geq \Pr[Q^t_{vw}|A^t_w] $.
		
		\item $ w \in R^{uv} $. In this case, we do not simplify $ \Pr[D^t_{uv}|A^t_w] $ any further.
		
		\item $ w \in \{u, v\} $. $ \Pr[D^t_{uv}|A^t_w] = 1 $ because the pivot is clustered in iteration $ t $.
	\end{enumerate}
	
	Hence,
	\begin{align*}
		\sum_{w \in V_t} \Pr[D^t_{uv}|A^t_w] \geq&\; \sum_{w \in T^{uv}_1} 1 && \text{Case 1} \\
		&\;+ \sum_{w \in T^{uv}_{2u}} \Pr[\overline{Q^t_{uw}}|A^t_w] && \text{Case 2} \\
		&\;+ \sum_{w \in T^{uv}_{2v}} \Pr[\overline{Q^t_{vw}}|A^t_w] && \text{Case 3} \\
		&\;+ \sum_{w \in S^{uv}_{2u} \setminus Y^{uv}_u} \Pr[Q^t_{uw}|A^t_w] && \text{Case 4} \\
		&\;+ \sum_{w \in S^{uv}_{2v} \setminus Y^{uv}_v} \Pr[Q^t_{vw}|A^t_w] && \text{Case 5} \\
		&\;+ \sum_{w \in R^{uv}} \Pr[D^t_{uv}|A^t_w] && \text{Case 6} \\
		&\;+ 2 && \text{Case 7}
	\end{align*}
	
	In order to prove the claim, it is sufficient to show that
	\begin{align*}
		\sum_{w \in V_t} E[M^t_{uv}|A^t_w] \leq&\; \max\left(2, \frac{3}{1+2p}\right) \sum_{w \in V_t} \Pr[D^t_{uv}|A^t_w]
	\end{align*}
	First, note the following obvious inequalities which follow by noting $\max\left(2, \frac{3}{1+2p}\right) \geq 2$. Let $\alpha=\max\left(2, \frac{3}{1+2p}\right)$.

	$$\;\sum_{w \in R^{uv}} \Pr[D^t_{uv}|A^t_w] \\ \;\leq \alpha \sum_{w \in R^{uv}} \Pr[D^t_{uv}|A^t_w] $$
	\begin{align*}&\;\frac{3}{2} \sum_{w \in T^{uv}_{2u}} \Pr[\overline{Q^t_{uw}}|A^t_u] + \frac{3}{2} \sum_{w \in T^{uv}_{2v}} \Pr[\overline{Q^t_{vw}}|A^t_v] \\ \;&\leq \alpha \left(\sum_{w \in T^{uv}_{2u}} \Pr[\overline{Q^t_{uw}}|A^t_w] + \sum_{w \in T^{uv}_{2v}} \Pr[\overline{Q^t_{vw}}|A^t_w]\right) 
	\end{align*}
	\begin{align*}&\;2 \sum_{w \in S^{uv}_{2u} \setminus Y^{uv}_u} \Pr[Q^t_{uw}|A^t_u] + 2 \sum_{w \in S^{uv}_{2v} \setminus Y^{uv}_v} \Pr[Q^t_{vw}|A^t_v] \\ &\;\leq \alpha \left(\sum_{w \in S^{uv}_{2u} \setminus Y^{uv}_u} \Pr[Q^t_{uw}|A^t_w] + \sum_{w \in S^{uv}_{2v} \setminus Y^{uv}_v} \Pr[Q^t_{vw}|A^t_w]\right)
	\end{align*}

	Therefore, considering the remaining terms of $\sum_{w \in V_t} E[M^t_{uv}|A^t_w]$ and $\sum_{w \in V_t} \Pr[D^t_{uv}|A^t_w]$, to establish $\sum_{w \in V_t} E[M^t_{uv}|A^t_w] \leq\; \max\left(2, \frac{3}{1+2p}\right) \sum_{w \in V_t} \Pr[D^t_{uv}|A^t_w]$, it remains to be shown that
	
	\begin{align*}
		&\;\Pr[Q^t_{uv}|A^t_u] + \Pr[Q^t_{uv}|A^t_v] \\
		&+ |T^{uv}_1| + 2|T^{uv}_1|(1-\Pr[Q^t_{uv}|A^t_u]) \\ &\;\leq \alpha(|T^{uv}_1| + 2)
	\end{align*}

	Since $ 1-(1-p)^{|T^{uv}_1|} \leq \Pr[Q^t_{uv}|A^t_u] = \Pr[Q^t_{uv}|A^t_v] \leq 1 $, it suffices to show that

	\begin{align*}
		2+|T^{uv}_1|+2|T^{uv}_1|(1-p)^{|T^{uv}_1|} \leq&\; \max\left(2, \frac{3}{1+2p}\right)(|T^{uv}_1|+2)
	\end{align*}

	First, we show that when $ |T^{uv}_1| \geq 3 $, $ [1+2(1-p)^{|T^{uv}_1|}]|T^{uv}_1| \leq \frac{3}{1+2p} |T^{uv}_1| $. Equivalently, we want to show that when $ x \geq 3 $, $ [1+2(1-p)^{x}] \leq \frac{3}{1+2p} $. Suppose for contradiction that $ 1+2(1-p)^x > \frac{3}{1+2p} $.
	\begin{align*}
		1+2(1-p)^{x} >&\; \frac{3}{1+2p} \\
		2(1-p)^{x} >&\; \frac{2-2p}{1+2p} \\
		(1-p)^{x} >&\; \frac{1-p}{1+2p} \\
		(1-p)^{x-1}(1+2p) >&\; 1 \\
	\end{align*}
	Since $ x \geq 3 $, we have that
	\begin{align*}
		(1-p)^2(1+2p) >&\; 1 \\
		(1-2p+p^2)(1+2p) >&\; 1 \\
		1-4p^2 + p^2 + 2p^3 >&\; 1 \\
		2p-3 >&\; 0 \\
		p >&\; \frac{3}{2}
	\end{align*}
	We have reached a contradiction. Finally, we show that if $ x \leq 2 $, then $ [1+2(1-p)^{x}]x + 2 \leq \frac{3}{1+2p}(x+2) $. If $ x = 0 $, then the claim is clearly true. When $ x = 1 $,
	\begin{align*}
		1+2(1-p)+2 - \frac{3}{1+2p} \cdot 3 =&\; 5-2p - \frac{9}{1+2p} \\
		=&\; \frac{(5-2p)(1+2p)-9}{1+2p} \\
		=&\; \frac{-4p^2+8p-4}{1+2p} \\
		=&\; \frac{-4(p-1)^2}{1+2p} \leq 0
	\end{align*}
	When $ x = 2 $,
	
	\begin{align*}
		&[1+2(1-p)^2] \cdot 2 + 2 - \frac{3}{1+2p} \cdot 4 \\
		=&\; 4 + 4(1-p)^2 - \frac{12}{1+2p} \\
		=&\; \frac{4+8p+4(1-p)^2(1+2p)-12}{1+2p} \\
		=&\; \frac{4+8p+4(1-4p^2+p^2+2p^3)-12}{1+2p} \\
		=&\; \frac{8p-12p^2+8p^3-4}{1+2p} \\
		=&\; \frac{4(2p-3p^2+2p^3)-4}{1+2p}
	\end{align*}

    The derivative of $ f(p) = \frac{4(2p-3p^2+2p^3)-4}{1+2p} $ is strictly positive on $ p \in [0, 1] $, so the function is maximized at $ f(1) = 0 $.

	Thus, we have shown that $ E[ALG^t|V_t] \leq \max\left(2, \frac{3}{1+2p}\right) E[OPT^t|V_t] $. By Lemma ~\ref{lemma:loop}, the claim follows.
\end{proof}

\subsection{Proof of Lemma~\ref{thm:lowerbound}}
We now prove Lemma~\ref{thm:lowerbound}. Our proof approach follows that of \cite{Komusiewicz-11}, which proves a similar theorem based on the ETH. We first state a lemma from \cite{Komusiewicz-11}, and for completeness, we provide their proof of the lemma. For each vertex $ u $, we denote by $ N^+(u) $ the set of vertices with which $ u $ has $ + $-edges and by $ N^-(u) $ the set of vertices with which $ u $ has $ - $-edges.
\begin{lem}
	\label{lem:lowerbound}
	(Lemma 2.1 in \cite{Komusiewicz-11}) Let $ G = (V, E) $ be the input graph for a correlation clustering instance. There exists an optimal solution to the correlation clustering instance defined by $ G $ s.t. for any two vertices $ u $ and $ v $ for which $ \{u, v\} $ is a $ - $ edge and $ |N^+(u) \cap N^+(v)| \leq 1 $, $ u $ and $ v $ are in different clusters.
\end{lem}
\begin{proof}
	Consider any optimal solution $ S $ to the correlation clustering instance defined by $ G $. We will construct an optimal solution $ S' $ from $ S $ that satisfies the desired property. Consider an arbitrary but particular pair of vertices $ u $ and $ v $ s.t. $ u $ and $ v $ are in the same cluster $ C $ in $ S $ but $ \{u, v\} $ is not a $ + $ edge and $ |N^+(u) \cap N^+(v)| \leq 1 $. Suppose without loss of generality that $ |N^+(v) \cap C \setminus N^+(u)| \geq |N^+(u) \cap C \setminus N^+(v)| $. Note that the collection $ N^+(v) \cap C \setminus N^+(u) $, $ N^+(u) \cap C \setminus N^+(v) $, $ \{u\} $, $ \{v\} $, $ N^+(u) \cap N^+(v) \cap C $ partitions $ C $. If we were to remove $ u $ from $ C $ and create a singleton cluster containing $ u $ in $ S' $, we would require one more mistake on a $ + $ edge for each vertex in $ N^+(u) \cap C \setminus N^+(v) $ and for each vertex in $ N^+(u) \cap N^+(v) \cap C $ and we would require one less mistake on a $ - $ edge for each vertex in $ N^+(v) \cap C \setminus N^+(u) $ and for $ v $. Since $ |N^+(v) \cap C \setminus N^+(u)| \geq |N^+(u) \cap C \setminus N^+(v)| $ and since $ |\{v\}| = 1 \geq |N^+(v) \cap N^+(u)| $, we have that the extra cost incurred to $ S' $ over $ S $ due to this modification is $ |N^+(u) \cap C \setminus N^+(v)|+|N^+(u) \cap N^+(v)| - |N^+(u) \cap C \setminus N^+(v)| - |\{v\}| \leq 0 $.
	
	Repeating this process iteratively until all pairs satisfy the desired property, we will obtain an optimal solution $ S' $ because for each iteration a distinct singleton cluster is formed and at most $ n $ such singleton clusters can be formed.
\end{proof}
\begin{proof}[Proof of Lemma \ref{thm:lowerbound}]
	The following proof is based on the proof of Theorem 2.1 in \cite{Komusiewicz-11}.
	
	Consider the gap version of 3-SAT. In this problem, we are given $ n' $ boolean variables and $ m $ CNF clauses with at most $ 3 $ literals each, and we are asked to decide whether all clauses are simultaneously satisfiable or whether fewer than $ (1-\gamma) $ of the clauses are simultaneously satisfiable, where $ \gamma $ is as defined in Hypothesis ~\ref{hyp:gapeth}. We want to reduce this problem to the correlation clustering problem.
	
	Before describing the construction of the input graph $ G $ constructed from the input instance of gap 3-SAT, we ensure that any given literal appears in a given clause at most once (if it appears more than once, we remove all but one occurrence). After this preprocessing, let $ X $ be the set of all variables, and let $ Y $ be the set of all clauses. Let $ n = |X| $, and note that $ m = |Y| $. Note that the maximum proportion of the clauses that can be simultaneously satisfied is unchanged by this preprocessing.
	
	We now describe the construction of the input graph $ G = (V, E) $ for correlation clustering. Any edge that is not explicitly labeled as $ + $ according to our description below is labeled as $ - $. For each $ x \in X $, let $ c(x) $ be the number of clauses that include $ x $. For each $ x \in X $, create $ 4c(x) $ vertices in $ V $ and label the edges between each pair among the $ 4c(x) $ vertices so that the $ + $ edges form a cycle. For each $ x \in X $, we number the vertices in $ x $'s cycle of $ + $ edges from $ 1 $ to $ 4c(x) $ in a manner consistent with the order of the vertices in the cycle. For each $ x \in X $, let $ v_{x, i} $ denote the $ i $th vertex in $ x $'s cycle. Next, for each $ x \in X $, we number the clauses that include $ x $ from $ 1 $ to $ c(x) $ according to an arbitrary permutation. Note that if the literals $ x $ and $ \overline{x} $ both occur in some single clause $ y $, these two literals contribute separately to $ c(x) $ (so the clause $ y $ will receive two separate indices in the aforementioned permutation). For each $ x \in X $ and each $ y \in Y $ s.t. $ y $ includes $ x $, let $ \pi_x(y) $ be the index assigned to $ y $ in $ x $'s permutation of clauses. For all $ + $ edges in variable cycles, we say that the edge is even if it is of the form $ \{v_{x, 2j}, v_{x, 2j+1}\} $ for some integer $ j $ and odd otherwise. If $ 12m-4\sum_{x \in X} c(x) > 0 $, then we create a ``dummy'' variable cycle containing $ 12m-4\sum_{x \in X} c(x) > 0 $ vertices, and we add a dummy variable $ v $ to $ X $ with $ c(v) = 3m-\sum_{x \in X \setminus \{v\}} c(x) $. Thus, we now have that $ \sum_{x \in X} c(x) = 3m $.
	
	For each $ y \in Y $, we create a vertex $ v_y $ in $ V $. For each variable $ x $ that appears in $ y $, if the literal $ x $ appears in $ y $, then we label as $ + $ the edges $ \{v_{x, 4\pi_x(y)-3}, v_y\} $ and $ \{v_{x, 4\pi_x(y)-2}, v_y\} $. For each variable $ x $ that appears in $ y $, if the literal $ \overline{x} $ appears in $ y $, then we label as $ + $ the edges $ \{v_{x, 4\pi_x(y)-2}, v_y\} $ and $ \{v_{x, 4\pi_x(y)-1}, v_y\} $. Also, if the clause $ y $ contains both $ x $ and $ \overline{x} $ for some $ x \in X $, then we create a new vertex $ v_a $ and label as $ + $ the edge $ \{v_y, v_a\} $. If the clause $ y $ contains only one literal, then we create four new vertices $ v_{a_1} $, $ v_{a_2} $, $ v_{b_1} $, and $ v_{b_2} $, and we label as $ + $ the edges $ \{v_{a_1}, v_y\} $, $ \{v_{a_2}, v_y\} $, $ \{v_{b_1}, v_y\} $, and $ \{v_{b_2}, v_y\} $. Finally, if the clause $ y $ contains exactly two literals, then we create two new vertices $ v_{a_1} $ and $ v_{a_2} $, and we label as $ + $ the edges $ \{v_{a_1}, v_y\} $ and $ \{v_{a_2}, v_y\} $. An example of a ``clause gadget'' is shown in Figure 1.
	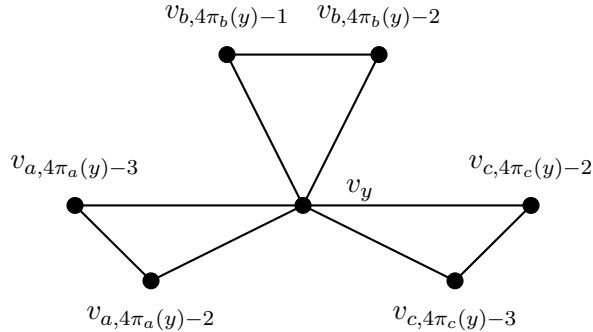
\begin{figure}[ht!]
		\begin{center}
			\begin{tikzpicture}
			%% vertices
			\draw[fill=black] (-0.25,0) circle (3pt);
			\draw[fill=black] (-1.25, 1) circle (3pt);
			\draw[fill=black] (3.75,0) circle (3pt);
			\draw[fill=black] (4.75, 1) circle (3pt);
			\draw[fill=black] (1.75,1) circle (3pt);
			\draw[fill=black] (0.75,3) circle (3pt);
			\draw[fill=black] (2.75, 3) circle (3pt);
			%% vertex labels
			\node at (-0.25,-0.5) {$ v_{a, 4\pi_a(y)-2} $};
			\node at (-1.25, 1.5) {$ v_{a, 4\pi_a(y)-3} $};
			\node at (3.75,-0.5) {$ v_{c, 4\pi_c(y)-3} $};
			\node at (4.75, 1.5) {$ v_{c, 4\pi_c(y)-2} $};
			\node at (2.5,1.2) {$ v_y $};
			\node at (0.75, 3.5) {$ v_{b, 4\pi_b(y)-1} $};
			\node at (2.75, 3.5) {$ v_{b, 4\pi_b(y)-2} $};
			%%% edges
			\draw[thick] (-1.25, 1) -- (-0.25,0) -- (1.75,1) -- (-1.25, 1);
			\draw[thick] (0.75, 3) -- (2.75, 3) -- (1.75, 1) -- (0.75, 3);
			\draw[thick] (3.75, 0) -- (4.75, 1) -- (1.75, 1) -- (3.75, 0);
			\end{tikzpicture}
			\caption{Clause gadget for clause $ y = a \vee \overline{b} \vee c $}
		\end{center}
	\end{figure}
	
	First, we will show that $ 10m $ is a lower bound on the optimal correlation clustering cost. We say that vertices $ u $ and $ v $ are $ + $-neighbors if $ \{u, v\} $ is labeled as $ + $, and for each vertex $ u $ we say that the set of all $ + $-neighbors of $ u $ is the $ + $-neighborhood of $ u $. Since for any two vertices that are not $ + $-neighbors the intersection of their $ + $-neighborhoods is of size at most $ 1 $, it follows from Lemma ~\ref{lem:lowerbound} that there exists an optimal solution to correlation clustering on $ G $ that consists only of mistakes on $ + $-edges. Note that any optimal clustering must make a mistake on every other $ + $-edge. In each variable cycle, there are $ 4c(x) $ edges and since $ \sum_{x \in X} c(x) = 3m $, $ 4 \cdot 3m \frac{1}{2} = 6m $ mistakes on $ + $ edges in the variable cycles must be made. Next, consider each clause gadget. It is easy to verify that at least $ 4 $ mistakes must be made on the $ + $ edges incident on the clause vertex. Hence, $ 6m + 4m = 10m $ is a lower bound on the optimal cost.
	
	Now we show that when the 3-SAT instance is satisfiable, there exists an optimal solution to the correlation clustering problem on $ G $ that has a cost of $ 10m $. Fix a satisfying assignment $ \mathcal{A} $ for the 3-SAT instance. For each variable $ x $, if $ x $ is true in $ \mathcal{A} $, then we choose to make mistakes on the even $ + $ edges in $ x $'s variable cycle, and otherwise we choose to make mistakes on the odd $ + $ edges in $ x $'s variable cycle. Now consider an arbitrary but particular clause vertex $ v_y $. By construction of $ G $, there is at least one $ + $ edge $ \{u, w\} $ from a variable cycle that participates in $ v_y $'s clause gadget and on which we do not make a mistake in the previous step. We make mistakes on all $ 4 $ $ + $ edges incident on $ v_y $ other than $ \{v_y, u\} $ and $ \{v_y, w\} $. The subgraph consisting of only $ + $ edges then consists only of 3-cycles or paths of length $ 2 $, so we have a valid clustering. Moreover, we have made exactly $ 6m + 4m = 10m $ edge modifications.
	
	It remains to be shown that when fewer than $ (1-\gamma)m $ clauses can be satisfied in the input 3-SAT instance, any valid solution to the correlation clustering problem requires more than $ \left(1+\frac{\gamma}{10}\right)10m $ mistakes. As argued above, all mistakes made by the optimal solution are on $ + $ edges. In particular, for each variable $ x_i \in X $, any optimal clustering of $ G $ will either make mistakes on all even $ + $ edges or all odd $ + $edges in $ x_i $'s cycle in $ G $. Next, fix a clause $ y \in Y $, and let $ v_y \in V $ be the vertex corresponding to $ y $. Let $ \ell_1 $, $ \ell_2 $, and $ \ell_3 $ be the literals in $ y $, and let $ e_1 $, $ e_2 $, and $ e_3 $ be the $ + $ edges in the variable cycles to which $ v_y $ is connected; $ e_1 $, $ e_2 $, and $ e_3 $ correspond to the literals $ \ell_1 $, $ \ell_2 $, and $ \ell_3 $. If an optimal solution makes mistakes on $ e_1 $, $ e_2 $, and $ e_3 $, then it must make mistakes on all $ 6 $ $ + $ edges incident on $ v_y $. On the other hand, if the optimal clustering does not make a mistake on at least $ 1 $ of $ e_1 $, $ e_2 $, and $ e_3 $, then the clustering must make a mistake on $ 4 $ $ + $ edges incident on $ v_y $ (so that the two remaining $ + $ edges form a $ (+, +, +) $ triangle with one of $ e_1 $, $ e_2 $, or $ e_3 $). Note that if the clause $ y $ contains fewer than $ 3 $ literals or if $ y $ contains $ x $ and $ \overline{x} $ for some $ x \in X $, there are still $ 6 $ edges incident on $ v_y $, but we may not need to make a mistake on all $ + $ edges incident on $ v_y $ when we make mistakes on all of the variable cycle $ + $ edges in $ v_y $'s gadget because $ v_y $ could be in a cluster with one of the extra vertices added to the gadget (i.e. one that is not in a variable cycle). However, we must make a mistake on at least $ 5 $ edges incident on $ v_y $ to eliminate all $ (+, +, -) $ triangles. Moreover, $ 4 $ mistakes on $ + $ edges are still required when the solution does not make a mistake on one of the variable cycle $ + $ edges in $ v_y $'s gadget.
	
	We will now show that the optimal number of mistakes in $ G $ is greater than $ 10m\left(1+\frac{\gamma}{10}\right) $. Suppose for contradiction that this claim were false. Fix an optimal clustering $ \mathcal{C} $. Let $ \alpha $ be the proportion of clauses $ C $ such that in the gadget corresponding to $ C $, $ \mathcal{C} $ makes a mistake on at least $ 5 $ $ + $ edges. Then by the argument in the previous paragraph, the number of mistakes made by $ \mathcal{C} $ is at least $ 6m + 4(1-\alpha)m + 5\alpha m = 10m + \alpha m $. Since $ 10m + \alpha m \leq 10m\left(1+\frac{\gamma}{10}\right) $, $ \alpha \leq \gamma $. Now consider the assignment $ \mathcal{A} $ for the 3-SAT instance in which a variable $ x $ is set to True iff the odd edges in $ x $'s cycle in $ G $ are deleted by $ \mathcal{C} $. Note that $ \mathcal{A} $ satisfies every clause $ C $ whose corresponding vertex $ v_C \in V $ is clustered in a triangle by $ \mathcal{C} $. Since there are $ (1-\alpha)m \geq (1-\gamma)m $ such clauses, $ \mathcal{A} $ satisfies at least $ (1-\gamma)m $ of the clauses. We have reached a contradiction.
	
	Note that the number of mistakes of an optimal clustering in $ G $ is a linear function of $ m $. Let $ N = |V| \leq \sum_{x \in X} 4c(x) + m + 4m = 12m + 5m = 17m $ (the variable cycles contribute $ \sum_{x \in X} 4c(x) $ vertices, the clause vertices contribute $ m $ vertices, and the extra vertices that may be added in clause gadgets add at most $ 4m $ vertices). Assume for contradiction that there were a $ \left(1+\frac{\gamma}{10}\right) $-approximation algorithm $ \mathcal{A} $ for correlation clustering on $ N $ vertices that runs in time $ 2^{o(C_{OPT})}\mbox{poly}(N) $. Consider an arbitrary but particular gap 3-SAT instance with $ n $ variables and $ m $ clauses. Running the reduction above on this instance would require time polynomial in $ n $ and $ m $. Furthermore, we note that in the resulting correlation clustering instance, $ C_{OPT} \leq |E^+| = 12m+6m = 18m $ because one valid clustering (in which all vertices are singletons) can be obtained by making mistakes on all $ + $ edges in the graph (there are $ 12m $ edges in the variable cycles and $ 6m $ edges in the clause gadgets. So in the correlation clustering instance, $ C_{OPT} \leq 18m $ and $ N \leq 17m $, so we would have an algorithm for gap 3-SAT that runs in time $ 2^{o(m)}\mbox{poly}(m) \in 2^{o(m)} $. Such an algorithm would contradict Hypothesis ~\ref{hyp:gapeth}.
\end{proof}
\end{document}